\documentclass[sigconf,nonacm]{acmart}
\usepackage{algorithm}
\usepackage{algpseudocode}
\usepackage{todonotes}
\usepackage{subcaption}
\usepackage{amsmath,amsfonts}
\usepackage{graphicx}
\usepackage{textcomp}
\newtheorem{definition}{Definition}
\newtheorem{example}{Example}
\newtheorem{lemma}{Lemma}

\newtheorem{theorem}{Theorem}

\usepackage{epstopdf}
\usepackage{booktabs}
\usepackage{url}
\usepackage{bbm}
\usepackage{multirow}
\usepackage{enumitem}
\usepackage{color}
\newcommand{\kwnospace}[1]{{\ensuremath {\mathsf{#1}}}}
\newcommand{\stitleunderline}[1]{\vspace{1ex}\noindent\underline{{\bf #1}}}
\newcommand{\stitle}[1]{\vspace{0.5ex} \noindent{\bf #1}}
\long\def\comment#1{}

\begin{document}

\title{$\delta$-EMG: A Monotonic Graph Index for Approximate Nearest Neighbor Search}

\author{Liming Xiang}
\affiliation{%
  \institution{Beijing Institute of Technology}
  \city{Beijing}
  \country{China}}
\email{3120240967@bit.edu.cn}

\author{Jing Feng}
\affiliation{%
	\institution{Beijing Institute of Technology}
	\city{Beijing}
	\country{China}}
\email{scererymc@gmail.com}

\author{Ziqi Yin}
\affiliation{%
	\institution{Nanyang Technological University}
	\country{Singapore}}
\email{ziqi003@e.ntu.edu.sg}

\author{Zijian Li}
\affiliation{%
	\institution{Huawei Technologies Ltd}
	\country{China}}
\email{lizijian8@huawei.com}

\author{Daihao Xue}
\affiliation{%
	\institution{Huawei Technologies Ltd}
	\country{China}}
\email{xuedaihao@huawei.com}

\author{Hongchao Qin}
\affiliation{%
	\institution{Beijing Institute of Technology}
	\city{Beijing}
	\country{China}}
\email{qhc.neu@gmail.com}

\author{Ronghua Li}
\affiliation{%
	\institution{Beijing Institute of Technology}
	\city{Beijing}
	\country{China}}
\email{lironghuabit@126.com}

\author{Guoren Wang}
\affiliation{%
	\institution{Beijing Institute of Technology}
	\city{Beijing}
	\country{China}}
\email{wanggrbit@gmail.com}

\begin{abstract}
Approximate nearest neighbor (ANN) search in high-dimensional spaces is a foundational component of many modern retrieval and recommendation systems. Currently, almost all algorithms follow an $\epsilon$-Recall-Bounded principle when comparing performance: they require the ANN search results to achieve a recall of more than $1-\epsilon$ and then compare query-per-second (QPS) performance. However, this approach only accounts for the recall of true positive results and does not provide guarantees on the deviation of incorrect results. To address this limitation, we focus on an Error-Bounded ANN method, which ensures that the returned results are a $(1/\delta)$-approximation of the true values.
Our approach adopts a graph-based framework. To enable Error-Bounded ANN search, we propose a $\delta$-EMG (Error-bounded Monotonic Graph), which, for the first time, provides a provable approximation for arbitrary queries. By enforcing a $\delta$-monotonic geometric constraint during graph construction, $\delta$-EMG ensures that any greedy search converges to a $(1/\delta)$-approximate neighbor without backtracking. Building on this foundation, we design an error-bounded top-$k$ ANN search algorithm that adaptively controls approximation accuracy during query time. To make the framework practical at scale, we introduce $\delta$-EMQG (Error-bounded Monotonic Quantized Graph), a localized and degree-balanced variant with near-linear construction complexity. We further integrate vector quantization to accelerate distance computation while preserving theoretical guarantees.
Extensive experiments on the ANN-Benchmarks dataset demonstrate the effectiveness of our approach. Under a recall requirement of 0.99, our algorithm achieves 19,000 QPS on the SIFT1M dataset, outperforming other methods by more than 40\%. Additional results show that, under the same precision requirements, our $\delta$-Error-bounded approach achieves higher QPS than all existing SOTA methods.
\end{abstract}



\maketitle

\section{Introduction}
Nearest neighbor search on high-dimensional vector data is a fundamental problem and has many real-world applications, including image retrieval \cite{imagesearch1,imagesearch2}, large language models \cite{llm1,llm2} and recommendation system \cite{recommend2}. In simple terms, given a dataset \(V \subset \mathbb{R}^d\) and a query point \(q \in \mathbb{R}^d\), the goal of top-\(k\) NN search is to find the \(k\) points in \(V\) closest to \(q\) based on a distance metric, typically the Euclidean distance. To improve search efficiency, much of the current research focuses on Approximate Nearest Neighbor (ANN) search, which trades a small loss in accuracy for significant gains in speed. Nowadays, most ANN search methods rely on the \(\epsilon\)-Recall-Bounded principle, requiring a recall of more than \(1-\epsilon\) and comparing query-per-second (QPS) performance. However, this approach does not ensure error bounds for the returned results.

To enable error-bounded ANN Search, researchers have conducted various studies.
From a theoretical perspective, Har-Peled's work showed that an approximate Voronoi diagram can achieve a \((1+\epsilon)\)-approximation for top-1 NN queries in logarithmic time, though it requires space on the order of \(O(n/\epsilon^d)\)~\cite{avd}. Arya et al. later improved this by reducing the ANN problem to an approximate polytope membership query, achieving near-optimal space complexity \(O(n/\epsilon^{d/2})\) while maintaining logarithmic query time~\cite{avd2,avd3}. In the database field, methods like NSG~\cite{nsg} can guarantee finding the exact top-1 NN if the query point is part of the dataset (\(q \in V\)). Similarly, FANNG~\cite{fanng} and \(\tau\)-MG~\cite{taumg}  offer exact top-1 NN guarantees only when the distance to the true nearest neighbor is within a predefined threshold \(\tau\). While these approaches contribute to error-bounded ANN search, they are limited to top-1 NN problems, which is insufficient for practical applications.\textbf{ In this paper, we propose a graph-based method that achieves a \((1/\delta)\)-approximation for top-\(k\) nearest neighbor search with a space complexity of \(O(n \ln n)\).} To the best of our knowledge, this is the {first method} to provide an error-bounded guarantee for top-\(k\) (\(k > 1\)) ANN search.

Besides, to improve search efficiency, current mainstream methods involve vector quantization to accelerate distance calculations. Among these, the state-of-the-art method, SymphonyQG\cite{symqg}, has integrated RaBitQ\cite{rabitq} into a proximity graph. We also incorporate our proposed error-bounded proximity graph into quantization and introduce an exploration strategy that adaptively selects between approximate and exact distance computations. This significantly reduces the errors that can arise from using only approximate distances. \textbf{Experimental results show that our error-bounded approach achieves higher QPS than all existing state-of-the-art methods under the same precision requirements.}

To achieve efficient error-bounded ANN search, we first introduce the \(\delta\)-Error-Bounded Monotonic Graph ($\delta$-$\kwnospace{EMG}$). This is a proximity graph where a monotonic search ensures that the final node's distance from any query \(q\) is no more than $1/\delta$ times the true distance. We define the Occlusion Region on $\delta$-$\kwnospace{EMG}$, using it to select neighbors and construct the graph. During ANN search on this graph, finding a local optimum node (a node with no closer neighbors to query \(q\)) ensures the search is Error-Bounded. Furthermore, to enhance performance, we propose a quantized version of the proximity graph, \(\delta\)-EMQG, together with an exploration strategy that balances approximate and exact distance computations.

{Our contributions are summarized as follows:}

\stitleunderline{Models.} We introduce the $\delta$-Error-Bounded Monotonic Graph ($\delta$-$\kwnospace{EMG}$), a novel proximity graph model designed to provide formal, rank-aware search guarantees. The core promise of this model is that for any query $q$, the $i$-th returned neighbor is a provable $(1/\delta)$-approximation of the true $i$-th nearest neighbor, under the practical condition that a local optimum node is found during the search. This guarantee is upheld by the graph's fundamental navigational property: for any starting point, a simple monotonic search is guaranteed to lead into a target region around the query, which we term the $\delta$-neighborhood. This prevents the search from getting trapped in poor local optimum node far from the query, thus ensuring the error bound is met.

\stitleunderline{Algorithms.} 
We propose a suite of algorithms for constructing and searching the $\delta$-$\kwnospace{EMG}$. To realize the aforementioned error guarantees, we design a graph construction method based on a unique geometric occlusion rule. Our exact implementation of this rule yields a graph with an expected out-degree of $O(\log n)$ and $O(n \log n)$ space complexity. However, its $O(n^2 \log n)$ time complexity makes it impractical for large datasets.

To overcome this, we develop a highly efficient approximate construction algorithm that reduces the time complexity to near-linear (\(O(L n^{(d + 1)/d} \ln(n^{1/d}) / \Delta)\)) and the space complexity to $O(n)$. This method produces a graph that serves as a localized and degree-balanced approximation of a true $\delta$-$\kwnospace{EMG}$. Furthermore, we present a new search method for the $\delta$-$\kwnospace{EMG}$ that dynamically determines the candidate set size to ensure the desired error bounds.

Finally, we integrate vector quantization to create the $\delta$-$\kwnospace{EMQG}$, a quantized variant that utilizes RaBitQ \cite{rabitq} for highly efficient distance estimation. This is paired with a novel probing search strategy that intelligently balances fast, approximate distance calculations with precise, exact ones to maximize throughput while maintaining high accuracy.

\stitleunderline{Experiments.} We perform experiments on six real-world datasets to show the effectiveness and efficiency of our proposed methods. The results are compared with the existing methods such as NSG, HNSW, NGT-QG and SymphonyQG. Our algorithm achieves a single-thread QPS of 19,000 on the SIFT1M dataset. At 99\% recall with $k=1$, it outperforms the best baseline by 1.2x to 3.2x. The error analysis matches our theoretical proof, confirming the accuracy of our method. Additionally, we conduct scalability and ablation experiments to further validate the algorithm.

\section{Related Work}
The ANN search problem has been widely explored, resulting in various indexing and search strategies~\cite{survey1, survey2, survey3, annbenchmark}. Below, we categorize the existing work into Non-Graph-Based and Graph-Based approaches and introduce them separately.

\subsection{Non-Graph-Based ANN Search}
Traditional ANN methods fall into several major categories.
Tree-based methods \cite{rtree,mtree,kdtree} recursively divide the data space to support hierarchical search.
While effective in low-dimensional settings, their performance deteriorates sharply in high dimensions due to the curse of dimensionality.  
Hashing-based methods \cite{hash1,hash2,hash3} map nearby points to the same hash bucket with high probability, enabling sublinear search in expectation.  
Inverted file (IVF) based methods\cite{ivf} cluster the dataset and build an inverted index from cluster centroids to data points, drastically reducing the search space to a few promising clusters.
Quantization-based methods, such as Product Quantization (PQ) \cite{pq,opq,fastscan}, compress high-dimensional vectors by partitioning them into subspaces and quantizing each subspace independently.

\subsection{Graph-Based ANN Search}
Graph-based methods represent the state-of-the-art in ANN search. Given a set of vectors $V$, they construct a graph $G=(V,E)$ where nodes correspond to vectors and edges define neighborhood relationships among them. Such graphs are collectively referred to as proximity graphs \cite{pg}. Queries are answered by traversing the graph, typically using a greedy search, to find the approximate nearest neighbors.
Existing graph-based approaches can be broadly grouped into three main families.

\stitle{$k$-NN Graph Variants.}
$k$-Nearest Neighbor ($k$-NN) graph connects each node to its $k$ nearest neighbors. It serves as a practical approximation of the theoretically ideal but computationally expensive Delaunay Graph (DG)\cite{voronoi, voronoi_is_monotonic}. 
Unlike the DG, which becomes nearly dense in high dimensions, the $k$-NN graph remains sparse. It can be efficiently built using methods like NN-Descent, which exhibits an empirical complexity of about $O(N^{1.14})$ \cite{nndescent}.
The $k$-NN graph serves as the foundational structure for many well-known ANN search algorithms, including GNNS\cite{gnns}, IEH\cite{ieh} and DPG\cite{dpg}.

\stitle{Navigable Small-World (NSW) Models.}
Inspired by the small-world phenomenon\cite{nsw_exp}, this family of methods constructs graphs with both short- and long-range links to facilitate efficient greedy routing\cite{nsw_thm}. The Navigable Small World (NSW) algorithm~\cite{nsw} pioneered a practical, data-agnostic approach to building these graphs.
Hierarchical NSW (HNSW) algorithm~\cite{hnsw} extended NSW with a hierarchical multi-layer structure that supports coarse-to-fine traversal—rapidly covering large distances in sparse upper layers and refining locally in the dense base layer. To date, HNSW remains one of the most effective and widely adopted ANN algorithms, though it does not provide formal accuracy guarantees.

\stitle{RNG-Based Methods.}
RNG\cite{rng} is constructed by eliminating the longest edge in every triangle of points. This principle yields highly sparse graphs with a constant average degree\cite{pg}. While the pure RNG is not sufficient for effective navigation\cite{msnet}, it inspired the Monotonic RNG (MRNG)\cite{nsg}, a variant that guarantees finding the exact nearest neighbor for in-dataset queries via greedy search.
However, this guarantee fails for out-of-dataset queries. To address this, algorithms like FANNG\cite{fanng} and $\tau$-Monotonic Graph ($\tau$-MG) \cite{taumg} were developed, which provide exact search guarantees for queries that are very close (within a radius $\tau$) to the dataset, but reverting to heuristic behavior otherwise. LMG \cite{lmg} extends the guarantee of $\tau$-MG by preserving edges for a range of $\tau$ values. Other works such as Vamana\cite{diskann} and SSG\cite{ssg} have proposed heuristic modifications to the RNG rules to improve search performance. 

Additionally, significant effort has been invested in optimizing ANN search for modern hardware, leading to specialized systems for SSDs \cite{diskann, fresh_diskann, spann} and GPUs \cite{song,gpu}. Unlike previous work, ours is the first to offer an error-bounded guarantee for top-\(k\) (\(k > 1\)) ANN search and achieves high QPS during searches.

\section{Preliminaries}

In this section, we formally define the nearest neighbor search problem and establish the notation used throughout this paper. Default notations are listed in Table \ref{tab:notations}.

\begin{table}[t!]
	\centering
	\small
	\vspace{-3mm}
	\caption{Notations used in this paper}
	\vspace{-3mm}
	\label{tab:notations}
	\begin{tabular}{@{} l p{0.65\columnwidth} @{}}
		\toprule
		\textbf{Notation} & \textbf{Description} \\
		\midrule
		$\mathbb{R}^d$ & The $d$-dimensional real vector space. \\
		$V = \{v_1, v_2, \dots, v_n\}$  & A dataset of $n$ vectors in $\mathbb{R}^d$. \\
		$d(x, y), \|x - y\|$ & Euclidean distance between vectors $x, y$. \\
		$\tilde{d}(x, y)$ & Approximate Euclidean distance between $x, y$. \\
		$G=(V, E)$ & A Proximity Graph with vertices $V$ and edges $E$. \\
		$\mathcal{N}(v)$ & The set of neighbors of node $v$ in the graph $G$. \\
		$q$ & A query vector in $\mathbb{R}^d$. \\
		$N_k(q)$ & The exact top-$k$ nearest neighbors of $q$. \\
		$R_k(q)$ & The results of top-$k$ ANN search of $q$. \\
		$C[i]$ & The $i$-th nearest vector to $q$ in the set $C$. \\
		$C[1{:}l]$ & The subset of $C$ containing the $l$ nearest vectors to $q$ (or all of $C$ if $|C|<l$). \\
		\bottomrule
	\end{tabular}
\end{table}

\subsection{\underline{N}earest \underline{N}eighbor {\small (NN)} Search}

Let $V = \{v_1, v_2, \dots, v_n\}$ be a dataset of $n$ vectors in $\mathbb{R}^d$. The distance between any two vectors $x, y \in \mathbb{R}^d$ is measured by the Euclidean distance, denoted as $d(x,y) = \|x - y\|$.

\begin{definition}[top-$k$ NN Search]
	\label{def:knn_search}
	Given a dataset $V$ and a query vector $q \in \mathbb{R}^d$, the goal of top-$k$ Nearest Neighbor (NN) search is to find a set $N_k(q) \subset V$ of size $k$ such that for any vector $v \in N_k(q)$ and any vector $v' \in V \setminus N_k(q)$, the inequality $d(q,v) \leq d(q,v')$ holds. In cases of distance ties, the selection of $N_k(q)$ is arbitrary.
\end{definition}

For subsequent definitions, we treat $N_k(q)$ as a sequence $(v_{(1)}, \allowbreak v_{(2)}, \allowbreak \dots, v_{(k)})$ ordered by non-decreasing distance to $q$. Solving the exact $k$-NN problem via a linear scan has a computational cost of $O(nd)$, which is prohibitive for large-scale and high-dimensional datasets. This has motivated the development of approximate solutions that offer a trade-off between search accuracy and efficiency.

\subsection{\underline{A}pproximate \underline{N}earest \underline{N}eighbor {\small (ANN)} Search}

Due to the high computational cost of exact Nearest Neighbor search and the fact that real-world applications often do not require highly accurate k-nearest neighbors, most current research focuses on Approximate Nearest Neighbor (ANN) search. To facilitate performance comparisons, the field has largely standardized on an $\epsilon$-Recall-Bounded evaluation paradigm, which can be defined as follows:

\begin{definition}[top-$k$ $\epsilon$-Recall-Bounded ANN Search]
	\label{def:rbanns}
	Given a dataset $V$, a query vector $q \in \mathbb{R}^d$ and a parameter $\epsilon \in (0, 1)$, the goal of $\epsilon$-Recall-Bounded ANN Search is to return a set of $k$ vectors $R_k(q)$, that satisfies $recall(R_k(q)) = \frac{R_k(q)\cap N_k(q)}{|N_k(q)|} = \frac{R_k(q)\cap N_k(q)}{k}$, is at least $1-\epsilon$.
\end{definition}

While the $\epsilon$-Recall-Bounded metric is widely adopted in the literature \cite{nsg,hnsw,symqg}, its fundamental drawback is its exclusive focus on retrieving true positives. It does not guarantee the error margin in cases of mismatches. In this paper, we focus on a novel approach called $\delta$-Error-Bounded ANN search, which ensures that every returned result has an error bound relative to the ground truth neighbor.

\begin{definition}[top-$k$ $\delta$-Error-Bounded ANN Search]
	\label{def:ebanns}
	Given a dataset $V$, a query vector $q$, and an approximation factor $\delta \in (0, 1)$, the goal is to return a sequence of $k$ vectors, $R_k(q) = (r_{(1)}, r_{(2)}, \dots, r_{(k)})$, ordered by non-decreasing distance to $q$, that satisfies the following condition for all $i \in \{1, \dots, k\}$:
	$ d(q, r_{(i)}) \le (1/\delta) \cdot d(q, v_{(i)}) $.
\end{definition}

Our rank-aware definition is more practical than the stricter alternative of requiring $d(q, r_{(i)}) \le (1/\delta) \cdot d(q, v_{(1)}) $ for all $i$ \cite{nsg,ssg}, a condition that rapidly becomes infeasible as $k$ increases.

\stitle{Hardness of achieving $\delta$-Error-bounded ANN Search.}
Achieving the guarantee in Definition~\ref{def:ebanns} for arbitrary queries in high-dimensional space is a notoriously difficult problem. The core difficulty lies in the fact that most existing index structures are constructed based solely on the dataset $V$. Their geometric and navigational properties are therefore defined relative to the data points themselves, not for the entire continuous space $\mathbb{R}^d$.

Consequently, even the most advanced provable methods offer only limited guarantees, as their guarantees typically hinge on restrictive assumptions about the query's location relative to the dataset \cite{nsg,fanng,taumg}. They fall short of providing a universal error bound for an arbitrary 1-ANN query, let alone the rank-aware guarantee for $k$-ANN. Our work is primarily motivated by addressing this limitation.

\section{How to Achieve $\delta$-Error-bounded ANN Search}

\subsection{ANN Search on Proximity Graph}
Graph-based ANN search methods construct a proximity graph over the dataset $V$, where the search is performed via graph traversal. Algorithm~\ref{alg:greedy_search} outlines a common greedy strategy for this purpose.

The algorithm maintains a candidate set $C$ and a set $T$ of visited nodes. Vectors in $C$ are sorted by their ascending distance to the query $q$. We use the notation $C[1:l]$ to denote the top-$l$ candidates in $C$. Correspondingly, $C[i]$ denotes the $i$-th candidate in this sorted list. The algorithm iteratively selects the closest unvisited node from candidate set $C$, adds its neighbors to the set, and prunes the set to maintain the top-$l$ best candidates. This process terminates when all nodes within $C[1:l]$ have been visited.

\begin{algorithm}[t!]
	\caption{ANN search on a (Proximity) Graph}
	\label{alg:greedy_search}
	\begin{algorithmic}[1]
		\Require (Proximity) Graph $G=(V,E)$, query vector $q$, result size $k$, start node $v_s \in V$, candidate set size $l$
		\Ensure $R_k(q)$: a set of $k$ approximate nearest neighbors of $q$
		
		\State candidate set $C \gets \{v_s\}$; visited set $T \gets \varnothing$;
		\While{ $\exists u \in C[1{:}l], u \notin T$ }
		\State $u \gets \arg\min_{u \in C[1{:}l] \setminus T}d(q,u)$;
		\State $T \gets T \cup \{u\}$;
		\ForAll{$v \in \mathcal{N}(u) \setminus T$}
		\State $C \gets C \cup \{v\}$;
		\EndFor
		\State Prune $C$ to retain the top $l$ candidates closest to $q$;
		\EndWhile
		\State \Return $C[1{:}k]$;
	\end{algorithmic}
\end{algorithm}

This search algorithm is fundamentally a heuristic. Without a specially designed graph structure, there is no formal assurance that the search will not terminate in a poor local optimum, failing to find high-quality neighbors.

\subsection{The Proximity Graph to be Monotonic} 
To prevent searches from terminating in such local optima, the graph itself must possess stronger navigational properties. A key property is monotonicity, which guarantees that for any destination node, a path of progressively closer nodes always exists. This concept was originally introduced as the Monotonic Search Network (MSNET) in \cite{msnet}. To formalize this, we first define a monotonic path with respect to an arbitrary query.

\begin{definition}[Monotonic Path]
	\label{def:monopath}
	Given a directed graph $G=(V, E)$ defined on vector set $V \subset \mathbb{R}^d$, and a query $q \in \mathbb{R}^d$, a path $(v_1, v_2, \dots, v_k)$ in $G$ is monotonic with respect to $q$ if the distance to $q$ strictly decreases at every step. That is, for all $i \in \{1, \dots, k-1\}$, the condition $d(q, v_i) > d(q, v_{i+1})$ holds.
\end{definition}

Using this definition, we can now define a monotonic graph.

\begin{definition}[Monotonic Graph]
	\label{def:msnet}
	A directed graph $G=(V, E)$ is a monotonic graph if for any two distinct nodes $u, v \in V$, there exists a monotonic path from $u$ to $v$ in $G$ with respect to $v$.
\end{definition}

\begin{figure}[t]
	\centering
	\begin{subfigure}[b]{0.21\textwidth}
		\centering
		\includegraphics[width=\linewidth]{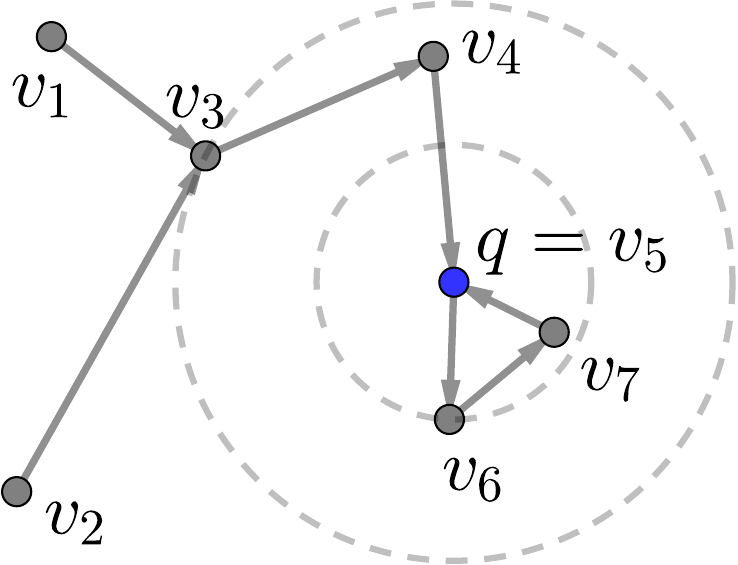}
		\caption{Monotonic Graph}
		\label{fig:msnet}
	\end{subfigure}
	\hspace{0.04\textwidth}
	\begin{subfigure}[b]{0.19\textwidth}
		\centering
		\includegraphics[width=\linewidth]{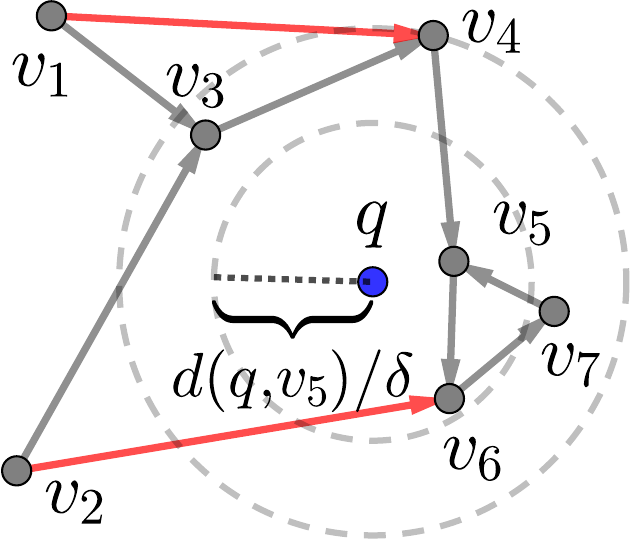}
		\caption{$\delta\text{-}\kwnospace{EMG}$}
		\label{fig:delta-msnet}
	\end{subfigure}
	\vspace{-3mm}
	\caption{A comparison of Monotonic Graph and $\delta\text{-}\kwnospace{EMG}$. $v_1, \dots, v_7$ are nodes in the graph and $q$ is the query. The red edges indicate additional links introduced in $\delta\text{-}\kwnospace{EMG}$ compared to the Monotonic Graph.}
	\label{fig:msnets}
\end{figure}

\begin{definition}[Monotonic Top-1 Search]
	\label{def:mono_search}
	Given a directed graph $G=(V, E)$, query vector $q$ and start node $v_s \in V$, a monotonic top-1 search iteratively moves from the current node $v_t$ to its neighbor closest to $q$:
	$$ v_{t+1} = \arg\min_{u \in \mathcal{N}(v_t)} d(q, u), \quad v_1=v_s $$
	where \( \mathcal{N}(v_t) \) denotes the neighbors of $v_t$.
	The process terminates when it reaches a node $v_f$ that is closer to $q$ than all of its neighbors. Such a node $v_f$ is called a local optimum node with respect to $q$.
\end{definition}

\begin{theorem}
	\label{thm:search_guaranteed}
	For any query $q \in V$, the monotonic top-1 search on a monotonic graph $G=(V, E)$ is guaranteed to terminate at $q$.
\end{theorem}

\begin{proof}
	Assume the search terminates at node $ v \neq q $.  
	Since $q \in V$ and $G$ is monotonic, there exists at least one monotonic path from $v$ to $q$.  
	Let $u$ be the immediate successor of $v$ on this path.  
	By Definition~\ref{def:monopath}, $d(q, u) < d(q, v)$, contradicting the termination condition that $v$ is a local optimum node.  
	Hence, the search must terminate at $v = q$.
\end{proof}

\begin{example}
	Figure~\ref{fig:msnet} illustrates the navigational property of a Monotonic Graph. The directed edges depict a subgraph that guarantees monotonic paths toward the destination node $q=v_5$. The definition of a monotonic graph guarantees that such a navigable structure exists between any pair of nodes. 
	However, consider a case where $q \ne v_5$, as shown in Figure~\ref{fig:delta-msnet}. In this case, using only the edges in the Monotonic Graph (gray edges in the figure), monotonic searches starting from $v_1$, $v_2$, or $v_3$ fail to reach $v_5$ because $v_3$ acts as a local optimum node.
\end{example}

Theorem~\ref{thm:search_guaranteed} provides a strong guarantee of finding the exact nearest neighbor, but it holds only when the query is a dataset point and only for the 1-NN case. Real-world scenarios, however, involve arbitrary queries that are not in the dataset and often require finding the top-$k$ neighbors.

This motivates the development of our framework, the $\delta$-Error-Bounded Monotonic Graph ($\delta\text{-}\kwnospace{EMG}$), which guarantees a top-$k$ $\delta$-error-bounded ANN search for arbitrary queries.

\subsection{Error-Bounded Monotonic Graph ($\delta\text{-}\kwnospace{EMG}$)}

This chapter introduces our primary contribution: a novel theoretical framework and a corresponding graph structure that provide provable error bounds for any arbitrary query $q \in \mathbb{R}^d$. We first define the $\delta$-Neighborhood for a query $q$ below.

\begin{definition}[$\delta$-Neighborhood]
	\label{def:delta_neighborhood}
	Given a dataset $V$ and a query $q$, let $v_{(1)} = \arg\min_{v \in V} d(q, v)$ be the true nearest neighbor of $q$. For a parameter $\delta \in (0, 1)$, the $\delta$-neighborhood of $q$ is defined as the closed ball centered at $q$ with radius $d(q, v_{(1)})/\delta$:
	$$ \{ x \in \mathbb{R}^d \mid d(q, x) \le \frac{1}{\delta} \cdot d(q, v_{(1)}) \} $$
\end{definition}

Based on this definition, we generalize the concept of $\delta$-error-bounded monotonic graph as follows.

\begin{definition}[$\delta$-Error-Bounded Monotonic Graph]
	\label{def:delta_emg}
	Given a dataset $V$ and $\delta \in (0, 1)$. A proximity graph $G=(V,E)$ is a $\delta$-error-bounded monotonic graph ($\delta\text{-}\kwnospace{EMG}$) if for any query $q \in \mathbb{R}^d$ and any starting node $v_s \in V$, there exists a monotonic path $(v_1, \dots, v_k)$ in $G$ such that $v_1 = v_s$ and $v_k$ lies within the $\delta$-neighborhood of $q$.
\end{definition}

\begin{example}
Figure \ref{fig:delta-msnet} illustrates this concept. The inner circle represents the $\delta$-neighborhood of $q$. While most nodes (e.g., $v_1, v_3, v_4, v_7$) have monotonic paths directly to the true nearest neighbor $v_5$, node $v_2$ has a monotonic path to $v_6$. Since $v_6$ lies within the $\delta$-neighborhood, the condition is satisfied for all nodes.
\end{example}

The navigational properties of a $\delta\text{-}\kwnospace{EMG}$ provide approximation guarantees for both top-1 ANN and top-$k$ ANN search.

\textbf{Error Bound for 1-NN:} $\delta\text{-}\kwnospace{EMG}$ guarantees to find a $(1/\delta)$-approximate nearest neighbor for any query.

\begin{theorem}
	\label{prop:delta-msnet-guarantee}
	For any query $q\in \mathbb{R}^d$, the monotonic top-1 search on a $\delta\text{-}\kwnospace{EMG}$ is guaranteed to return a node $r \in V$ that is a $(1/\delta)$-approximation of the true nearest neighbor $v_{(1)}$:
	$$ d(q,r) \leq \frac{1}{\delta} \cdot d(q,v_{(1)}) $$
\end{theorem}

\begin{proof}
	Let $r \in V$ be the node where the monotonic search terminates. By the Definition~\ref{def:mono_search}, $r$ is a local optimum node, meaning it has no neighbor in $G$ that is closer to $r$. Assume that the proposition is false, i.e., $d(q,r) > \frac{1}{\delta} \cdot d(q,v_{(1)})$. This would imply that $r$ lies outside the $\delta$-neighborhood of $q$. However, by the definition of a $\delta\text{-}\kwnospace{EMG}$, for any node outside the $\delta$-neighborhood, there must exist a monotonic path starting from it. The existence of such a path requires $r$ to have a neighbor $r'$ such that $d(q, r') < d(q, r)$. This directly contradicts the fact that $r$ is a local optimum node. Therefore, the assumption must be false, and we must have $d(q,r) \leq \frac{1}{\delta} \cdot d(q,v_{(1)})$.
\end{proof}

\textbf{Error Bound for $k$-NN:} For any query and our proposed search strategy in Section \ref{sec:anns_emg}, the $\delta\text{-}\kwnospace{EMG}$ ensures a top-$k$ $\delta$-error-bounded ANN search when there is a locally optimum node in the candidate set. This condition is frequently met in our experiments in Sec.\ref{sec:effect}. The detailed proof can be seen at Theorem \ref{thm:kann_guarantee} in Sec.\ref{sec:anns_emg}.

Based on the above properties, \textbf{the key challenge of achieving \(\delta\)-error-bounded ANN search has now shifted to the problem of constructing a $\delta\text{-}\kwnospace{EMG}$}. Below, we present the algorithms to construct this graph and to perform efficient ANN search upon it.

\section{Algorithms for $\delta\text{-}\kwnospace{EMG}$}

The core challenge lies in constructing a graph that satisfies the $\delta\text{-}\kwnospace{EMG}$ property. A naive approach would require verifying an infinite set of queries in $\mathbb{R}^d$ to ensure a monotonic path exists into each query's $\delta$-neighborhood, which is computationally intractable. The solution, therefore, is to translate this global guarantee into a practical, local rule that operates only on the finite data points in $V$. The following construction achieves this by introducing a carefully designed geometric occlusion condition.

\subsection{Construction of $\delta\text{-}\kwnospace{EMG}$}

We begin with the central element of our construction: a unique edge occlusion rule that determines the graph's topology.

\begin{definition}[Occlusion Region of $\delta\text{-}\kwnospace{EMG}$]
\label{def:delta_occlusion_region}
	For vectors $u, v \in \mathbb{R}^d$ and $\delta \in (0, 1)$, the Occlusion Region of $\delta\text{-}\kwnospace{EMG}$ is defined as:
	\begin{multline*}
		\text{Occlusion}_{\delta}(u,v) = \{ x \in \mathbb{R}^d \mid d(x, u) < d(u, v) \\
		\text{ and } d^2(x, v) + 2\delta \cdot d(u, v) \cdot d(x, u) < d^2(u,v) \}
	\end{multline*}
\end{definition}

\begin{figure}[t]\vspace{-3mm}
	\centering
	\includegraphics[width=0.45\columnwidth]{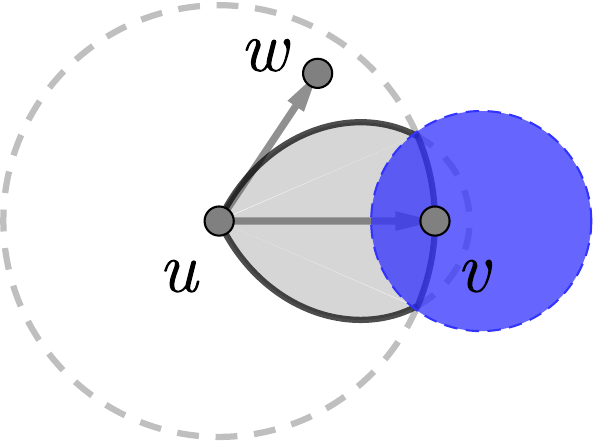}
	\vspace{-3mm}
	\caption{The construction rule for $\delta\text{-}\kwnospace{EMG}$. The gray shaded area represents the Occlusion Region for $(u,v)$. For any query $q$ inside blue shaded area, $u$ is guaranteed to have a neighbor closer to $q$.}
	\label{fig:occlusion}
\end{figure}

As illustrated in Figure \ref{fig:occlusion}, the occlusion region is the intersection of an open ball centered at $u$ and a teardrop-shaped volume with its cusp at $v$. The parameter $\delta$ modulates its geometry: as $\delta \to 0$, the region expands towards the lune of an MRNG. Conversely, as $\delta \to 1$, the region contracts significantly.

This specific geometry induces a "navigable region" around $v$ (the blue area in Figure~\ref{fig:occlusion}). For any query $q$ within this region, any vector $w$ in $\text{Occlusion}_{\delta}(u,v)$ is guaranteed to be closer to $q$ than $u$ is. The following lemma formalizes this.

\begin{lemma}
\label{lemma:region_property}
For any vectors $u, v$ and any $w \in \text{Occlusion}_{\delta}(u,v)$, any query $q \in \mathbb{R}^d$ that satisfies the condition $d(q,v) < \delta \cdot d(q,u)$, the inequality $d(q,w) < d(q,u)$ holds.
\end{lemma}

\begin{proof}
Without loss of generality, we can translate the coordinate system such that $u = \mathbf{0}$, where $\mathbf{0}$ is the zero vector.  The inequality to be proven, $d(q,w) < d(q,u)$ is equivalent to $\|q - w\|^2 < \|q\|^2$, which reduces to showing:
$$ 2q \cdot w > \|w\|^2 $$

First, we characterize the region containing the query $q$. The condition $d(q,v) < \delta \cdot d(q,u)$ becomes $\|q - v\| < \delta \|q\|$. Squaring both sides and expanding the expression yields:
$$ (1-\delta^2)\|q\|^2 - 2q \cdot v + \|v\|^2 < 0 $$

For $\delta \in (0, 1)$, completing the square with respect to $q$ gives:
$$ \left\|q - \frac{v}{1-\delta^2}\right\|^2 < \left(\frac{\delta\|v\|}{1-\delta^2}\right)^2 $$

This demonstrates that $q$ must lie within an open hypersphere $B(c, R)$ with center $c = \frac{v}{1-\delta^2}$ and radius $R = \frac{\delta\|v\|}{1-\delta^2}$.

To prove that $2q \cdot w > \|w\|^2$ for all $q \in B(c, R)$, it is sufficient to show that the inequality holds for the minimum value of the linear function $2q \cdot w$ over the ball. This minimum is attained on the boundary of the ball in the direction opposite to $w$, resulting in:
$$ \min_{q \in B(c,R)} (2q \cdot w) = 2(c \cdot w - R\|w\|) = \frac{2}{1-\delta^2} (v \cdot w - \delta \|v\| \|w\|) $$

Next, we use the condition on $w$. Since $w \in \text{Occlusion}_{\delta}(\mathbf{0},v)$, it must satisfy the second inequality from Definition \ref{def:delta_occlusion_region}:
$$ \|w-v\|^2 + 2\delta \|v\| \|w\| < \|v\|^2 $$

Expanding $\|w-v\|^2$ and simplifying this expression yields:
$$ 2(v \cdot w - \delta \|v\| \|w\|) > \|w\|^2$$

Finally, we substitute this result into our expression for the minimum of $2q \cdot w$:
$$ \min_{q \in B(c,R)} (2q \cdot w) > \frac{\|w\|^2}{1-\delta^2} \ge \|w\|^2 $$

This confirms that $2q \cdot w > \|w\|^2$ holds for any valid query $q$ in the navigable region, thus concluding the proof.
\end{proof}

We now give the construction of $\delta\text{-}\kwnospace{EMG}$ using the occlusion region.

\begin{theorem}
\label{thm:construct_of_deltaemg}
Given a dataset $V \subset \mathbb{R}^d$ and $\delta \in (0, 1)$, a directed graph $G=(V,E)$ is $\delta\text{-}\kwnospace{EMG}$ if for any pair of nodes $(u,v) \notin E$, there exists an edge $(u,w) \in E$ such that $w \in \text{Occlusion}_{\delta}(u,v)$. 
\end{theorem}

\begin{proof}
	Let $v_{(1)} = \arg\min_{v \in V} d(q, v)$ be the true nearest neighbor of a given query $q$. Consider any node $u \in V$ that lies outside the $\delta$-neighborhood of $q$, i.e., $d(q, u) > \frac{1}{\delta} \cdot d(q, v_{(1)})$.
	
	Now, consider the pair of nodes $(u, v_{(1)})$. According to the construction rule, there are two mutually exclusive possibilities for the edge $(u, v_{(1)})$:
	\begin{enumerate}
		\item The edge $(u, v_{(1)})$ exists in $E$. In this case, $v_{(1)}$ is a neighbor of $u$. From the initial condition, we have $d(q, u) > \frac{1}{\delta} \cdot d(q, v_{(1)}) > d(q, v_{(1)})$. Thus, $v_{(1)}$ is strictly closer to $q$ than $u$ is, and $(u, v_{(1)})$ is a monotonic step.
		
		\item The edge $(u, v_{(1)})$ does not exist in $E$. By theorem \ref{thm:construct_of_deltaemg}, its absence implies that it must be occluded by some other node $w$. This means there exists an edge $(u, w) \in E$ where $w \in \text{Occlusion}_{\delta}(u, v_{(1)})$. According to Lemma \ref{lemma:region_property}, we have $d(q, w) < d(q, u)$. Thus, $u$ has a neighbor $w$ that is strictly closer to $q$.
	\end{enumerate}
	
	In both cases, any node $u$ outside the $\delta$-neighborhood is guaranteed to have a neighbor in $G$ that is strictly closer to $q$. Therefore, any greedy search starting from $u$ must proceed along a monotonic path and cannot stop at $u$. Since the dataset $V$ is finite and the distance to $q$ strictly decreases at each step, this path must terminate. As it cannot terminate outside the $\delta$-neighborhood, it must eventually enter it. This satisfies the definition of a $\delta\text{-}\kwnospace{EMG}$.
\end{proof}

We analyze the expected out-degree of $\delta\text{-}\kwnospace{EMG}$, which is a key measure of the graph's structural complexity.

\begin{lemma}
	\label{lemma:deltaemg_deg}
	For a set $V$ of $n$ vectors sampled uniformly from a bounded region in $\mathbb{R}^d$, the expected out-degree of any vertex in the corresponding $\delta\text{-}\kwnospace{EMG}$ is $O(\ln n)$.
\end{lemma}

\begin{proof}
We assume that the dataset $V$ consists of $n$ points uniformly distributed in a bounded region of $\mathbb{R}^d$, with a constant point density $\rho$.
Let $P(r)$ denote the probability that an edge $(u,v)$ exists between two points $u,v \in V$ at distance $r = d(u,v)$.
Since the distribution is uniform and isotropic, $P(r)$ depends only on the distance $r$.

By Theorem~\ref{thm:construct_of_deltaemg}, an edge $(u,v)$ exists only if no other vertex $w$ creates an edge $(u,w)$ that occludes it, i.e.\ $w \in \text{Occlusion}_{\delta}(u,v)$.
For convenience, denote the occlusion region by
\[
\Omega_{\delta}(u,v) = \text{Occlusion}_{\delta}(u,v)
\]
The probability that $(u,v)$ survives occlusion can be written as
\[
P(r) = \mathbb{E}\!\left[\prod_{w \in V \cap \Omega_{\delta}(u,v)} (1 - P(d(u,w)))\right]
\]
Since the occlusion region expands as $r$ increases, it becomes increasingly likely to contain occluding vertices. Therefore, $P(r)$ is a monotonically decreasing function of $r$.

For any $w \in \Omega_{\delta}(u,v)$, by the definition of the occlusion region we have
\[
d^2(w,v) + 2\delta\, d(u,v)\, d(w,u) < d^2(u,v).
\]
Combining this with the triangle inequality $d(w,v) \ge |d(u,v) - d(w,u)|$,
and noting that $d(w,u) > 0$ since there are no duplicate points, we obtain
\[
d(u,w) < 2(1-\delta)\, d(u,v)
\]
Let $c = 2(1-\delta)$. Since $P(\cdot)$ is decreasing, for any such $w \in \Omega_{\delta}(u,v)$,
\[
P(d(u,w)) > P(c\, d(u,v)) = P(cr)
\]
Hence,
\[
P(r) < \mathbb{E}\!\left[\prod_{w \in V \cap \Omega_{\delta}(u,v)} (1 - P(cr))\right]
\]

Let the volume of the occlusion region be $\operatorname{Vol}(\Omega_{\delta}(u,v)) = C_{\delta,d}\, r^d$ for some constant $C_{\delta,d} > 0$ depending only on $\delta$ and $d$.
Under a Poisson point process approximation, the number of vertices inside the occlusion region follows a Poisson distribution with rate parameter
\[
\lambda = \rho \cdot \operatorname{Vol}(\Omega_{\delta}(u,v)) \cdot P(cr)
= \rho\, C_{\delta,d}\, r^d\, P(cr)
\]
The probability that this region is empty of occluding vertices is approximately $e^{-\lambda}$, giving
\[
P(r) < e^{-K\, r^d P(cr)},
\]
where $K = \rho\, C_{\delta,d} > 0$

We now show that the recursive inequality above implies that $P(r)$ decays at least as fast as $r^{-d}$.
Suppose, for contradiction, that there exist constants $\epsilon, \beta > 0$ such that
\[
P(r) > \beta\, r^{\epsilon - d}
\]
for all sufficiently large $r$
Substituting this into the inequality gives
\[
\beta\, r^{\epsilon - d}
< P(r)
< e^{-K\, r^d P(cr)}
< e^{-K\, r^d \beta (cr)^{\epsilon - d}}
= e^{-K'\, r^{\epsilon}},
\]
where $K' = K\, \beta\, c^{\epsilon - d} > 0$.
However, as $r \to \infty$, the polynomial term $\beta r^{\epsilon - d}$ decays much more slowly than the super-exponential term $e^{-K' r^{\epsilon}}$, which yields a contradiction.
Therefore, $P(r)$ must decay at least as fast as $O(r^{-d})$.

The expected out-degree of a vertex $u$ is
\[
\mathbb{E}[\deg^+(u)] \approx \int \rho\, P(r)\, dV
= \int_{r_{\min}}^{r_{\max}} \rho\, P(r)\, S_{d-1}\, r^{d-1}\, dr,
\]
where $S_{d-1}$ denotes the surface area of the unit $(d-1)$-sphere.
Using $P(r) = O(r^{-d})$, we obtain
\[
\mathbb{E}[\deg^+(u)] \le
O\!\left(\int_{r_{\min}}^{r_{\max}} r^{-d}\, r^{d-1}\, dr\right)
= O\!\left(\ln r \Big|_{r_{\min}}^{r_{\max}}\right)
\]

For a dataset of $n$ uniformly distributed points with fixed density $\rho$, the dataset volume scales as $O(n)$, so the maximum inter-point distance satisfies $r_{\max} = O(n^{1/d})$, while the minimum distance $r_{\min} = O(1)$.
Substituting these bounds yields
\[
\mathbb{E}[\deg^+(u)] = O(\ln r_{\max} - \ln r_{\min})
= O(\ln n).
\]
Thus, the expected out-degree of any vertex in a $\delta\text{-}\kwnospace{EMG}$ constructed from uniformly sampled points in $\mathbb{R}^d$ is $O(\ln n)$.
\end{proof}

Algorithm~\ref{alg:construct_delta_emg} provides a direct, exact implementation of the construction principle from Theorem~\ref{thm:construct_of_deltaemg}. For each node $u$, it considers all other nodes as potential neighbors, sorted by distance. It then iteratively adds an edge $(u,v)$ only if $v$ is not occluded by any of the shorter, already-accepted neighbors of $u$.

\begin{algorithm}[t!]
	\caption{Construct $\delta\text{-}\kwnospace{EMG}$}
	\label{alg:construct_delta_emg}
	\begin{algorithmic}[1]
		\Require Vector set $V$, parameter $\delta \in (0,1)$
		\Ensure A $\delta\text{-}\kwnospace{EMG}$ $G=(V,E)$
		
		\State Initialize edge set $E \gets \varnothing$
		\ForAll{node $u \in V$}
		\State $\mathcal{N}(u) \gets \Call{SelectNeighbors}{u, V, \delta}$ 
		\State $E \gets E \cup \{(u,v)\mid v\in\mathcal{N}(u)\}$
		\EndFor
		\State \Return $G = (V, E)$
		
		\Statex
		\Function{SelectNeighbors}{$u, V, \delta$} 
		\State $\mathcal{N}(u) \gets \varnothing$
		\ForAll{$v \in V$ sorted by increasing distance from $u$}
		\If{$\nexists w \in \mathcal{N}(u) \text{ and } w \in \text{Occlusion}_{\delta}(u,v)$}
		\State $\mathcal{N}(u) \gets \mathcal{N}(u) \cup \{v\}$
		\EndIf
		\EndFor
		\State \Return $\mathcal{N}(u)$
		\EndFunction
	\end{algorithmic}
\end{algorithm}

\stitle{Complexity of Algorithm~\ref{alg:construct_delta_emg}.} Recall that Lemma~\ref{lemma:deltaemg_deg} guarantees that $\delta\text{-}\kwnospace{EMG}$ has an expected out-degree of $O(\ln n)$.
Algorithm~\ref{alg:construct_delta_emg} requires space of $O(n \ln n)$. 
Algorithm~\ref{alg:construct_delta_emg} requires examining all pairs of points to test occlusion relationships. Consequently, its overall time complexity remains $O(n^2 \ln n)$.

\subsection{ANN Search on $\delta\text{-}\kwnospace{EMG}$}
\label{sec:anns_emg}

While the $\delta\text{-}\kwnospace{EMG}$ provides a foundational top-1 ANN guarantee, practical applications demand both the retrieval of $k$ nearest neighbors and the ability to specify a more stringent error bound at query time. The standard approach is to employ greedy search with a heuristically chosen candidate set size $l > k$, offering no formal connection between $l$ and the resulting error.

To address this, we introduce an error-bounded $k$-ANN search algorithm that operates on a $\delta\text{-}\kwnospace{EMG}$, presented as Algorithm \ref{alg:eb_ann_search}. 
Rather than using a fixed candidate set size, our algorithm adaptively expands the candidate set size $l$, starting from $l=k$. The search terminates when the distance to the $k$-th candidate is sufficiently close to the query relative to the current search frontier. This termination is governed by a user-specified parameter $\alpha \ge 1$.

\begin{algorithm}[t!]
	\caption{Error-Bounded top-$k$ ANN Search on $\delta\text{-}\kwnospace{EMG}$}
	\label{alg:eb_ann_search}
	\begin{algorithmic}[1]
		\Require $\delta\text{-}\kwnospace{EMG}$ Graph $G=(V,E)$, query $q$, start node $v_s \in V$, result size $k$, accuracy parameter $\alpha \ge 1$
		\Ensure $R_k(q)$: $k$ approximate nearest neighbors of $q$
		
		\State candidate set $C \gets \{v_s\}$, visited set $T \gets \varnothing$
		\For{candidate set size $l = k, k+1, k+2, \dots$}
		\While{$\exists u \in C[1{:}l], u \notin T$}
		\State $u \gets \arg\min_{u \in C[1{:}l] \setminus T}d(q,u)$
		\State $T \gets T \cup {u}$
		\ForAll{$v \in \mathcal{N}(u) \setminus T$}
		\State $C \gets C \cup \{v\}$
		\EndFor
		\State keep top $l+1$ candidates in $C$ in ascending distance
		\EndWhile
		\State \textbf{if} $d(q,C[l]) \ge \alpha \cdot d(q,C[k]) $ \textbf{then break}
		\EndFor
		
		\State \Return $C[1{:}k]$
	\end{algorithmic}
\end{algorithm}

The formal guarantee of this algorithm depends on the discovery of a local minimum during the search (defined as a node with no neighbors closer to $q$). Let $R_k(q) = (r_{(1)}, \dots, r_{(k)})$ be the result returned by the algorithm, and let $N_k(q) = (v_{(1)}, \dots, v_{(k)})$ be the true $k$ nearest neighbors, both ordered by distance to $q$.

\begin{theorem}
	\label{thm:kann_guarantee}
	Let $C$ be the final candidate set upon termination of Algorithm~\ref{alg:eb_ann_search}, and let $R_k(q)$ be the returned result. If $\exists u \in C \setminus R_k(q) $ and $u$ is a local optimum node, then $R_k(q)$ satisfies the following error bound for all $i \in \{1, \dots, k\}$:
	$$ d(q,r_{(i)}) \leq \frac{1}{\delta'} \cdot d(q,v_{(i)}) \quad \text{where} \quad \delta' = \delta \cdot \frac{d(q,u)}{d(q,r_{(k)})} $$
\end{theorem}

\begin{proof}
	Since $u$ is a local optimum node, it has no neighbor closer to $q$. By the property of the $\delta\text{-}\kwnospace{EMG}$ in Theorem \ref{prop:delta-msnet-guarantee}, it must satisfy $\delta \cdot d(q,u) \leq d(q,v_{(1)}) $. From the definition of $\delta'$ in the theorem statement, we have $\delta \cdot d(q,u) = \delta' \cdot d(q,r_{(k)})$.
	Substituting this into the previous inequality gives:
	$$ \delta' \cdot d(q,r_{(k)}) \leq d(q,v_{(1)}) $$
	We know that for any $i \in \{1, \dots, k\}$, $d(q,r_{(i)}) \leq d(q,r_{(k)})$ and $d(q,v_{(1)}) \leq d(q,v_{(i)})$. Combining these inequalities yields the full chain:
	$$ \delta' \cdot d(q,r_{(i)}) \leq \delta' \cdot d(q,r_{(k)}) \leq d(q,v_{(1)}) \leq d(q,v_{(i)}) $$
	This proves the theorem.
\end{proof}

The effective error bound \(\delta'\), as derived from Theorem \ref{thm:kann_guarantee}, is determined after the search process. Its value depends on whether there exists a local optimum \(u\) in the final candidate set \(C\). This condition is reasonable because the sparsity of high-dimensional spaces makes local optima statistically likely to exist. 

Note that, the parameter \(\alpha\) provided by the user directly influences the search accuracy. A larger \(\alpha\) enforces a stricter stopping criterion, which requires the algorithm to expand its candidate set size. This wider search increases the chances of finding a local optimum \(u\) farther from the query, leading to a stronger approximation guarantee. Thus, the quality of the approximation is closely tied to the choice of \(\alpha\). 

Exp-6 and 7 in Sec. \ref{sec:effect} demonstrate how \(\alpha\) impacts the achieved error bound \(\delta'\). The results show that if \(\alpha > 2\), we can find such a local optimum node with a probability of over 95\%, and in this case, the algorithm is theoretically guaranteed to be error-bounded.

\section{Algorithms for Quantized $\delta\text{-}\kwnospace{EMG}$ ($\delta\text{-}\kwnospace{EMQG}$)}

While the exact $\delta\text{-}\kwnospace{EMG}$ provides a robust theoretical foundation, its direct construction is computationally intractable and likely to cause a highly non-uniform degree distribution. Nodes in dense regions accumulate excessive edges, leading to high search overhead, which is a common challenge for many proximity graph methods.

To overcome these limitations, we introduce a approximate construction of $\delta\text{-}\kwnospace{EMG}$. The core insight is to differentiate the roles of edges based on their length: long-range edges for coarse navigation can have relaxed guarantees, while short-range edges for fine-grained convergence must be robust. This motivates a principled relaxation where the parameter $\delta$ becomes an adaptive function of the edge length for any $u,v \in V$:
$$ \delta_t(u, v) = 1 - \frac{d(u,v)}{d(u,v_{(t)})} $$
where $v_{(t)}$ is the $t$-th nearest neighbor of $u$ in $V$. 

This adaptive definition creates a multi-scale graph structure. For long-range edges where $d(u,v) > d(u,v_{(t)})$, $\delta(u,v)$ becomes negative, the deterministic guarantee is relaxed in favor of probabilistic progress. As the search enters a query's neighborhood, the value of $\delta(u,v)$ increases towards 1. This approximation thus locally recovers the strong navigable properties of a $\delta\text{-}\kwnospace{EMG}$ with high-$\delta$.

Practically, it provides two major benefits:
\begin{enumerate}
	\item Locality. Long-range edges are pruned early, so each node only explores a small local subset during construction, reducing complexity to near-linear.
	\item Degree Balancing. Dense regions prune more aggressively, preventing degree explosion; sparse regions retain enough edges for connectivity.
\end{enumerate}

\begin{algorithm}[t!]
	\caption{Approximate Construction of $\delta\text{-}\kwnospace{EMG}$}
	\label{alg:approx_construct_deltaemg}
	\begin{algorithmic}[1]
		\Require Vector set $V$, max out-degree $M$, candidate set size $L$, neighborhood-scale parameter $t \le L$, iterations $I$
		\Ensure  An Approximate $\delta\text{-}\kwnospace{EMG}$ $G=(V,E)$
		
		\State Let $v_s$ be the approximate medoid of $V$.
		\State Initialize $G=(V, E)$ from a top-M approximate NN graph.
		\For{iter from 1 to $I$}
		\State New edge set $E_{new} \gets \varnothing$
		\ForAll{$u \in V$}
		\State $R_u \gets \Call{GreedySearch}{G, v_s, u, L, L}$
		\State $\mathcal{N}(u) \gets \Call{LocallySelectNeighbors}{u, R_u, t}$
		\If{$|\mathcal{N}(u)| > M$} 
		\State $\mathcal{N}(u) \gets$ the $M$ closest nodes in $\mathcal{N}(u)$
		\EndIf
		\State $E_{new} \gets E_{new} \cup \{(u,v)\mid v\in\mathcal{N}(u)\}$
		\EndFor
		\State $G \gets (V, E_{new})$
		\State Add reverse edges to $G$ within the degree $M$.
		\State Connect any nodes unreachable from $v_s$ to their nearest reachable neighbors, subject to the degree limit $M$.
		\EndFor
		
		\Statex
		\Function{LocallySelectNeighbors}{$u, R_u, t$} 
		\State $\mathcal{N}(u) \gets \varnothing$
		\State Let $r_{(t)}$ be the $t$-th closest node to $u$ in $R_u$
		\ForAll{$r \in R_u$ sorted by increasing distance from $u$}
		\State $\delta \gets 1 - \frac{d(u,r)}{d(u,r_{(t)})}$
		\If{$\nexists w \in \mathcal{N}(u) \text{ and } w \in \text{Occlusion}_{\delta}(u,r)$}
		\State $\mathcal{N}(u) \gets \mathcal{N}(u) \cup \{r\}$
		\EndIf
		\EndFor
		\State \Return $\mathcal{N}(u)$
		\EndFunction
		
		\State \Return $G = (V, E)$
	\end{algorithmic}
\end{algorithm}

The practical approximate construction of a $\delta\text{-}\kwnospace{EMG}$, detailed in Algorithm \ref{alg:approx_construct_deltaemg}, is an iterative process that refines an initial bootstrap graph (e.g., an approximate k-NN graph). Each iteration rebuilds the graph by first using beam search to generate local candidates for every node, and then applying our adaptive occlusion rule to prune these candidates into the final neighbor sets. This process typically converges to a high-quality graph within 3-4 iterations.

\stitle{Complexity of Algorithm~\ref{alg:approx_construct_deltaemg}.} The dominant computational cost of Algorithm~\ref{alg:approx_construct_deltaemg} lies in lines~6 and~7. Line~6 performs a greedy search on the current graph to obtain a local candidate set. 
According to the analysis in~\cite{nsg}, the complexity of such a search is approximately 
$O(L\, n^{1/d} \ln(n^{1/d}) / \Delta)$, 
where $\Delta$ is the smallest distance between any two distinct points in $V$. Line~7 processes at most $L$ candidates and checks occlusion against at most $M$ accepted neighbors, 
yielding a per-node complexity of $O(LM)$. 
Since both $M$ and $L$ are small constants relative to $n$, 
the overall time complexity per iteration is
$O(L n^{(d + 1)/d} \ln(n^{1/d}) / \Delta)$. 
Considering space complexity, since the out-degree of each node is hard-capped at a constant $M$, the space complexity of Algorithm~\ref{alg:approx_construct_deltaemg} is \({O}(Mn)={O}(n) \).

\subsection{Construction of $\delta\text{-}\kwnospace{EMQG}$}

Building upon the proposed approximate construction framework, we further incorporate vector quantization to reduce the cost of distance computations during ANN search. 
Specifically, we adopt the RaBitQ quantization scheme~\cite{rabitq}, as it provides an unbiased distance estimator with a rigorous theoretical error bound. 
Moreover, RaBitQ achieves an exceptionally high compression ratio, reducing vectors to as little as one bit per dimension. While prior work~\cite{symqg} has successfully integrated RaBitQ into a graph-based index, their approach relies on a heuristic NSG-like structure, lacking the geometric guarantees offered by our $\delta\text{-}\kwnospace{EMG}$ formulation.

The estimation of approximate distances from RaBitQ codes is performed using FastScan \cite{fastscan,quickadc} that leverages SIMD instructions. FastScan processes vectors in fixed-size batches, typically a multiple of the SIMD width (e.g., 32). Consequently, if the number of a node's neighbors is not a multiple of this batch size, the final batch incurs wasted computational cycles.

To ensure perfect alignment with FastScan's operational model, we first set the maximum out-degree $M$ to a multiple of the SIMD batch size. Recall that the parameter $t$ in Algorithm~\ref{alg:approx_construct_deltaemg} monotonically influences the resulting neighborhood size: a larger $t$ expands the local distance scale $d(u, v_{(t)})$, thereby relaxing the pruning condition and yielding more neighbors. During construction, for any node whose initial neighbor set (pruned with a default $t$) is smaller than $M$, we perform a binary search on the parameter $t$ within the candidate range $[1, L]$ to find the smallest value that produces a neighborhood of exactly size $M$.

The overall construction of $\delta\text{-}\kwnospace{EMQG}$ follows the iterative framework outlined in Algorithm \ref{alg:approx_construct_deltaemg}. Specifically, the neighborhood alignment step takes place after global connectivity has been established. Subsequently, for each node, we compute and store the RaBitQ codes for its entire neighborhood.

\subsection{Quantized Search on $\delta\text{-}\kwnospace{EMQG}$}

\begin{algorithm}[t!]
	\caption{Probing top-$k$ ANN Search on $\delta\text{-}\kwnospace{EMQG}$}
	\label{alg:probe_search}
	\begin{algorithmic}[1]
		\Require $\delta\text{-}\kwnospace{EMQG}$ $G=(V,E)$, query vector $q$, start node $v_s \in V$, result size $k$, accuracy parameter $\alpha \ge 1$
		\Ensure $R_k(q)$: $k$ approximate nearest neighbors of $q$
		
		\State candidate set $C_e \gets \{v_s\}$, $C_a \gets \varnothing$, visited set $T_e \gets \varnothing$, $T_a \gets \varnothing$
		\State $d_{last} \gets d(q,v_s)$
		
		\For{candidate set size $l = k, k+1, k+2, \dots$}
		\While{true}
		\State $u \gets \arg\min_{u \in C_e[1{:}l] \setminus T_e}d(q,u)$
		\State $w \gets \arg\min_{w \in C_a[1{:}l] \setminus T_a}\tilde{d}(q,u)$
		\If{$u$ is null and $w$ is null} \textbf{break}
		\ElsIf{$\Call{NeedProbing}{u, w, d_{last}}$}
		\State Compute exact distance $d(q,w)$
		\State $C_e \gets C_e \cup \{w\}, T_a \gets T_a \cup \{w\}$
		\State maintain $C_e$ in ascending exact distance
		\Else
		\State $d_{last} \gets d(q,u)$
		\State Compute approx. distances $\tilde{d}(q,v)$ for all $v \in \mathcal{N}(u)$
		\State $C_a \gets C_a \cup \mathcal{N}(u) \setminus T_a, T_e \gets T_e \cup \{u\}$
		\State maintain $C_a$ in ascending approximate distance
		\EndIf
		\EndWhile
		\State \textbf{if} $d(q,C_e[l]) \ge \alpha \cdot d(q,C_e[k])$ \textbf{then break}
		\EndFor
		
		\State \Return $C_e[1{:}k]$
		
		\Statex
		\Function{NeedProbing}{$u, w, d_{last}$}
		\If{$u$ is null} \Return true
		\ElsIf{$d(q, u) > d_{last}$ \textbf{and} $w$ is not null \textbf{and} $\tilde{d}(q,w) < d(q, u)$} 
		\State \Return true
		\Else
		\State \Return false
		\EndIf
		\EndFunction
	\end{algorithmic}
\end{algorithm}

Searching on a quantized graph introduces a critical trade-off: relying solely on approximate distances for navigation compromises accuracy, whereas frequent recourse to exact distance computations negates the performance benefits of quantization. To resolve this, we introduce the \textbf{Probing Search} algorithm (Algorithm \ref{alg:probe_search}). This method seamlessly integrates the speed of quantized exploration with the accuracy of exact verification, while maintaining the adaptive structure of Algorithm \ref{alg:eb_ann_search} with a dynamic increment candidate set size.

The algorithm maintains two candidate sets:
\begin{enumerate}[leftmargin=*]
	\item Exact Candidate Set ($C_e$): It stores candidates whose exact distances $d(q, \cdot)$ have been computed. It serves as the basis for the final result and the termination condition.
	\item Approximate Candidate Set ($C_a$): It is a set of candidates discovered during graph traversal, ordered by their approximate distances $\tilde{d}(q, \cdot)$.
\end{enumerate}

The search dynamically alternates between two operations: \textbf{Expansion} and \textbf{Probing}. The decision logic, which depends on the best unvisited candidates from both sets ($u \in C_e[1{:}l]$ and $w \in C_a[1{:}l]$) and the distance of the last expanded node ($d_{last}$), is designed to minimize exact distance computations by invoking them only when the search encounters a potential local optimum.
\begin{enumerate}[leftmargin=*]
	\item Expansion: This is the default operation. The algorithm expands from the most promising exact candidate $u$ by retrieving its neighbors, computing their approximate distances in a batch via FastScan, and inserting them into $C_a$.
	\item Probing: This is the operation when the search using exact distances stops improving—where \(u\) is farther from the query than \(d_{last}\)—and the candidate \(w\) from the approximate set looks better. The algorithm then "probes" this candidate $w$ by computing its exact distance and promoting it to $C_e$.
\end{enumerate}

\section{Experiments}
In this section, we conduct extensive experiments to evaluate the proposed algorithms. We implement seven different algorithms for comparison:

\noindent\textit{(i)} NSG\cite{nsg}. We use an efficient implementation of NSG provided in the open-source \textit{Glass} library.

\noindent\textit{(ii)} HNSW\cite{hnsw}. We also use the implementation from the \textit{Glass} library for consistency.

\noindent\textit{(iii)} $\tau$-MNG\cite{taumg}. Since no official implementation is available, we reimplemented $\tau$-MG based on the NSG source code.

\noindent\textit{(iv)} NGT-QG\cite{onng}. NGT-QG is a quantized graph method developed in the open-source NGT library. It integrates Product Quantization with a proximity graph.

\noindent\textit{(v)} SymphonyQG\cite{symqg}. SymphonyQG represents the latest advancement in quantized graph. We use the official open-source implementation released by the authors.

\noindent\textit{(vi)} $\delta\text{-}\kwnospace{EMG}$. 
$\delta\text{-}\kwnospace{EMG}$ is our proposed error-bounded monotonic graph framework. The index is constructed using Algorithm~\ref{alg:approx_construct_deltaemg}, and queries are performed using the error-bounded k-ANN search (Algorithm~\ref{alg:eb_ann_search}).

\noindent\textit{(vii)} $\delta\text{-}\kwnospace{EMQG}$. $\delta\text{-}\kwnospace{EMQG}$ extends $\delta\text{-}\kwnospace{EMG}$ by integrating RaBitQ-based vector quantization. It employs the Probing Search (Algorithm~\ref{alg:probe_search}) for queries.

\begin{table}[t]\vspace{-3mm}
	\small
	\centering
	\caption{Statistics of datasets}
	\vspace{-3mm}
	\label{tab:datasets}
	\begin{tabular}{c|c|c|c|c}
		\hline
		Dataset & dimension & \# of base & \# of queries & LID \\ \hline
		SIFT1M		&128 	&1,000,000		&10,000 	&9.3\\
		SIFT50M	&128 	&50,000,000		&10,000 	&9.3\\
		GIST		&960 	&1,000,000		&1,000 	&18.9\\
		MSong		&420 	&992,272		&200 	&9.5\\
		Crawl		&300 	&1,989,995		&10,000 	&15.7\\
		Deep1M		&256 	&1,000,000		&1,000 	&12.1\\
		\hline
	\end{tabular}
	\vspace{-3mm}
\end{table}

\stitle{Datasets.} Our evaluation is conducted on six widely-used real-world datasets that span diverse modalities, dimensionalities, and data scales: {SIFT1M}, {SIFT50M}, {GIST}, {MSong}, {Crawl}, and {Deep1M}.
These datasets are standard benchmarks in ANN research and have been adopted in prior studies~\cite{nsg,hnsw,fanng,symqg,annbenchmark,survey1} to evaluate both accuracy and efficiency. 
Specifically, SIFT1M and SIFT50M consist of 128-dimensional image descriptors; GIST contains holistic image features; MSong includes audio features; Crawl comprises word embeddings; and Deep1M contains deep visual features.
Table~\ref{tab:datasets} summarizes their key statistics, including the \textit{Local Intrinsic Dimensionality} (LID), which measures the local geometric complexity of each dataset. Datasets with higher LID values are generally harder for ANN search.

\stitle{Parameter settings.} 
For the baseline methods, construction parameters were set according to the recommendations in prior studies\cite{hnsw,nsg,taumg,symqg}. For our proposed $\delta\text{-}\kwnospace{EMG}$ and $\delta\text{-}\kwnospace{EMQG}$, we fixed the candidate set size at $L=1000$ and the number of refinement iterations at 3 across all datasets. The maximum out-degree $M$ was generally set to 64. $M$ was reduced to 32 for $\delta\text{-}\kwnospace{EMG}$ on MSong and Crawl, and for $\delta\text{-}\kwnospace{EMQG}$ on SIFT50M, as this value is sufficient for achieving high performance on these datasets.

At query time, we varied the primary search-time parameters for each algorithm to obtain its performance profile across different accuracy levels. For NSG, HNSW, $\tau$-MG, and SymphonyQG, this corresponds to adjusting the search candidate set size. The search procedure for NGT-QG involves two key parameters, the search radius and reranking size; we varied combinations of these two parameters to generate its accuracy-efficiency curve. For our $\delta\text{-}\kwnospace{EMG}$ and $\delta\text{-}\kwnospace{EMQG}$, we varied the accuracy parameter $\alpha$.

\stitle{Experimental settings.} 
All experiments were conducted on a server equipped with an AMD Ryzen Threadripper 3990X 64-Core Processor and 320 GB of RAM, running a Linux 4.4 kernel. The processor supports the AVX2 instruction set, which is leveraged by methods employing SIMD acceleration. All algorithms were implemented in C++ and compiled with GCC 13.3 using the `-Ofast' optimization flag. Index construction for all methods was parallelized using 32 threads. All search performance was measured on a single thread. Each reported Queries-Per-Second (QPS) value is the average of five runs.

\subsection{Performance V.S. the Baselines}
\stitle{Exp-1. QPS Comparison with the Baselines.}
We first compare the Queries-Per-Second (QPS) against recall for all methods across the six datasets. 
For each method, its build-time parameters were fixed according to the parameter settings described earlier, while its search-time parameters were varied to generate the QPS–Recall curves.
We tested with $k=1, 10, 100$ on all datasets, targeting a recall range of 0.9 to 0.995. 
Figure~\ref{fig:qps-recall-part} presents the comprehensive results. The complete results for all datasets are provided in the appendix (Figure~A.1). Both axes follow the logarithmic scaling used in ANN-Benchmarks \cite{annbenchmark}, with the x-axis logarithmic in $1-\text{recall}$ to emphasize the high-recall region. Curves that are higher and to the right are better. We note that NGT-QG results are missing for SIFT50M because it ran out of our memory.

\begin{figure}[t]
	\centering
	\includegraphics[width=0.8\columnwidth]{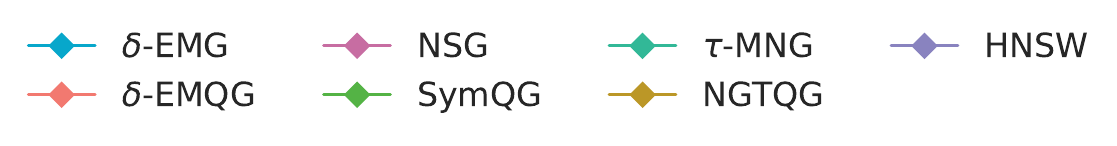}%
	\vspace{0.3em} 
	
	\captionsetup[subfigure]{justification=centering} 
	\begin{subfigure}[b]{0.50\columnwidth}
		\centering
		\includegraphics[width=\linewidth]{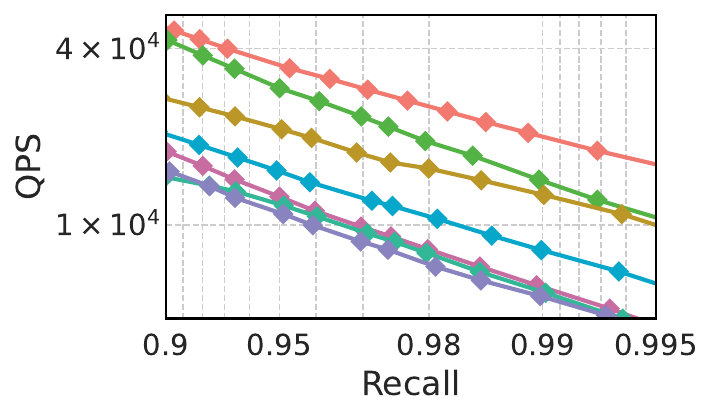}
		\caption{SIFT1M, k=10}
	\end{subfigure}
	\hspace{-0.03\columnwidth}
	\begin{subfigure}[b]{0.50\columnwidth}
		\centering
		\includegraphics[width=\linewidth]{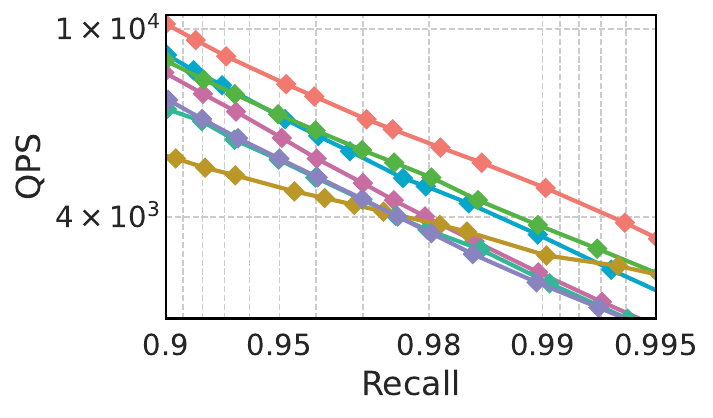}
		\caption{SIFT1M, k=100}
	\end{subfigure}
	
	\begin{subfigure}[b]{0.50\columnwidth}
		\centering
		\includegraphics[width=\linewidth]{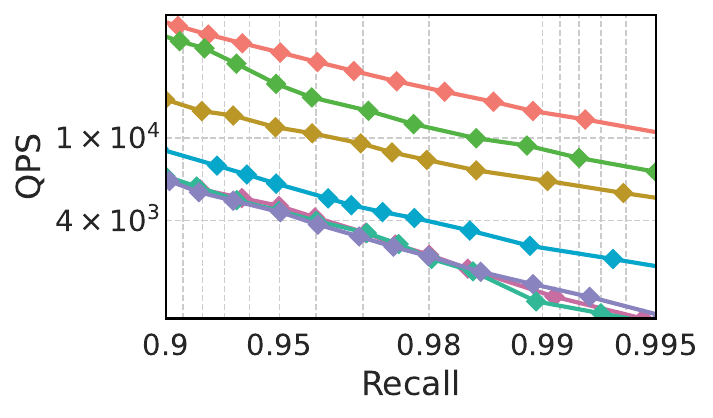}
		\caption{Deep1M, k=10}
	\end{subfigure}
	\hspace{-0.03\columnwidth}
	\begin{subfigure}[b]{0.50\columnwidth}
		\centering
		\includegraphics[width=\linewidth]{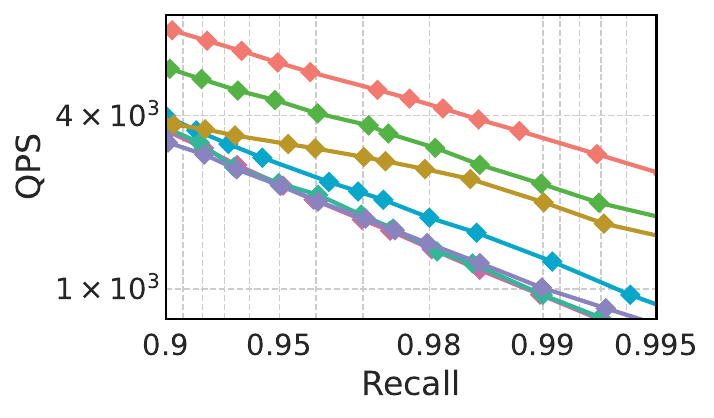}
		\caption{Deep1M, k=100}
	\end{subfigure}
	\vspace{-3mm}
	\caption{QPS vs Recall}
	\label{fig:qps-recall-part}
\end{figure}

The results clearly show that $\delta\text{-}\kwnospace{EMQG}$ outperforms all baseline methods across all datasets, k-values, and recall levels. At 99\% recall with $k=1$, $\delta\text{-}\kwnospace{EMQG}$ is 1.2x to 3.2x faster than the best baseline. For $k=10$ and $k=100$, this lead is between 1.2x and 2.1x.
Our non-quantized method, $\delta\text{-}\kwnospace{EMG}$, also exhibits highly competitive performance. It surpasses all other non-quantized baselines (HNSW, NSG, $\tau$-MG) in every tested configuration. For instance, at 99\% recall with $k=1$, $\delta\text{-}\kwnospace{EMG}$ is 2.1x to 2.3x faster than the best non-quantized baseline. Notably, at high recall regimes on the SIFT1M, SIFT50M, and MSong datasets, $\delta\text{-}\kwnospace{EMG}$ is even faster than the quantized baselines NGT-QG and SymphonyQG.

\textbf{Note that our algorithm delivers the best performance across all parameters and all datasets. For the complete results, please refer to the supplementary material.}

\stitle{Exp-2. Index Construction Comparison with the Baselines.}
We further analyze the index construction time and the final index size in memory. The index size includes the adjacency lists, vector storage, and quantization codes (for quantized methods). Figure~\ref{fig:construct-cost} shows these results.

\begin{figure}[t]\vspace{-3mm}
	\centering
	\includegraphics[width=0.76\columnwidth]{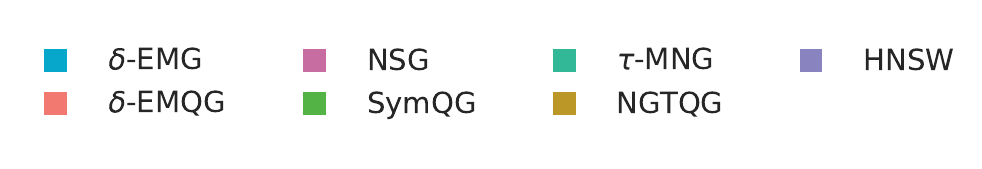}%
	\vspace{-0.3em} 
	
	\captionsetup[subfigure]{justification=centering} 
	\begin{subfigure}[b]{0.50\columnwidth}
		\centering
		\includegraphics[width=\linewidth]{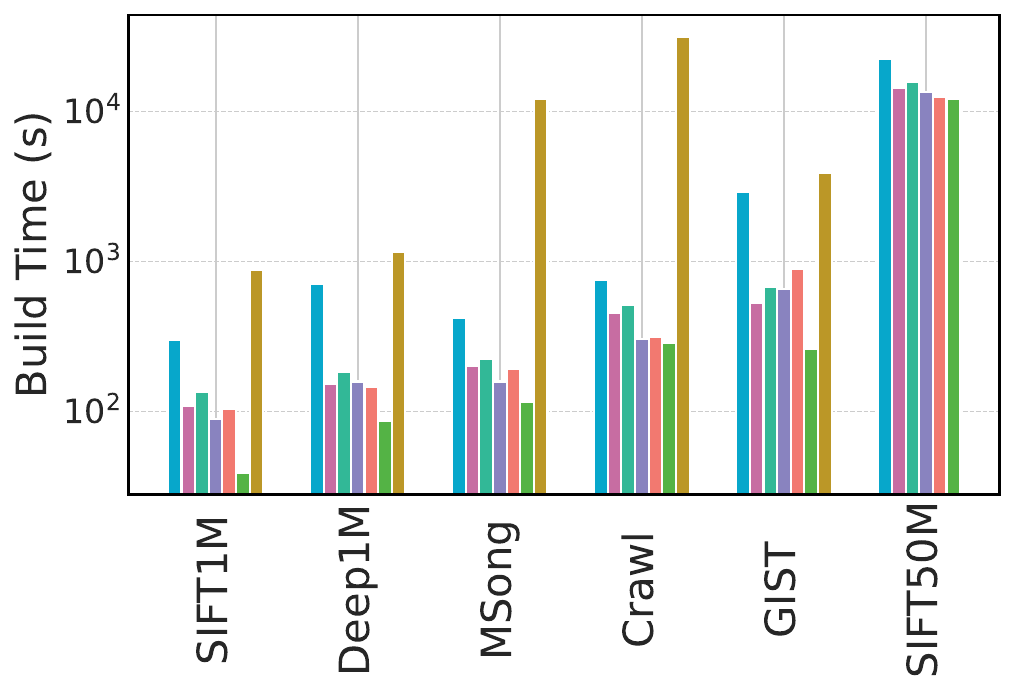}
		\caption{Build Time}
	\end{subfigure}
	\hspace{-0.02\columnwidth}
	\begin{subfigure}[b]{0.50\columnwidth}
		\centering
		\includegraphics[width=\linewidth]{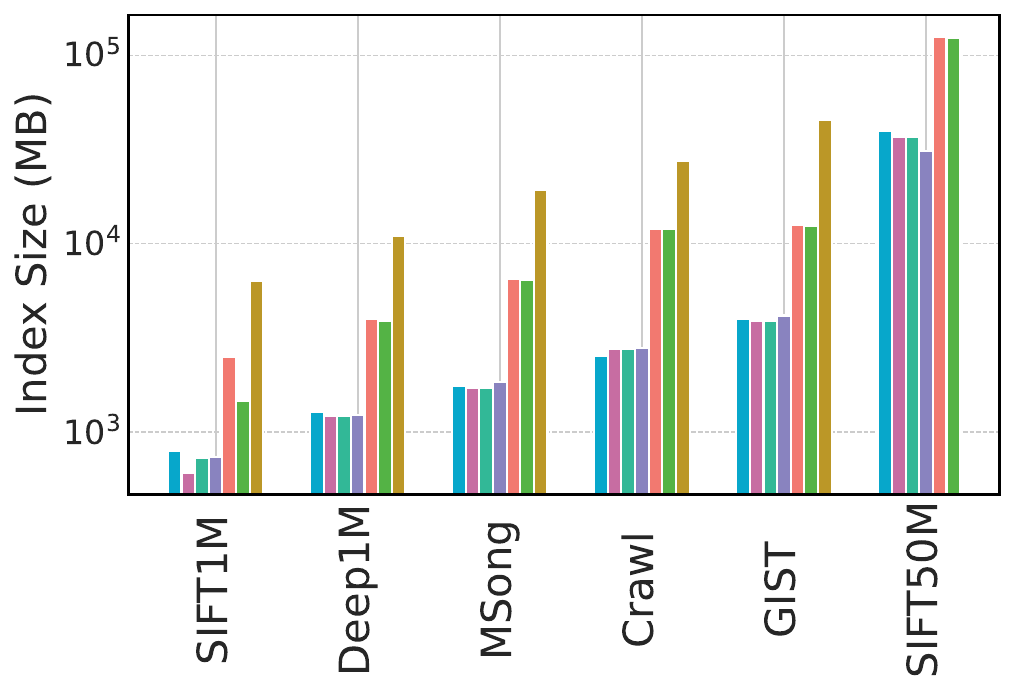}
		\caption{Index Size}
	\end{subfigure}
	\vspace{-6mm}
	\caption{Construction Cost}
	\vspace{-3mm}
	\label{fig:construct-cost}
\end{figure}

Overall, our proposed methods exhibit moderate index building costs in both time and space. The construction time for $\delta\text{-}\kwnospace{EMQG}$ is similar to the fastest baselines, HNSW and SymphonyQG. $\delta\text{-}\kwnospace{EMG}$ is slightly slower to build but still remains faster than NGT-QG.

In terms of index size, $\delta\text{-}\kwnospace{EMQG}$ has nearly the same size as SymphonyQG and smaller than NGT-QG. Meanwhile, the index size for $\delta\text{-}\kwnospace{EMG}$ is on par with the most space-efficient baselines like NSG and $\tau$-MG. This demonstrates that the superior search performance of our frameworks does not come at the cost of excessive indexing overhead.

\subsection{Effect of Construction Parameters}

\stitle{Exp-3. Effect of $\delta$.}
We investigated the impact of $\delta$ on search performance.
Since the exact construction of a $\delta\text{-}\kwnospace{EMG}$ (Algorithm~\ref{alg:construct_delta_emg}) is computationally prohibitive, we adopt the practical construction framework (Algorithm~\ref{alg:approx_construct_deltaemg}) with a fixed global $\delta$ throughout the entire build process, instead of using an adaptive value.

To isolate the effect of $\delta$, we varied its value while keeping all other construction parameters unchanged. For each resulting graph, we measured QPS at 95\% recall with $k=10$. The results are presented in Figure~\ref{fig:delta-sens-part}. Complete figures across all datasets are in the appendix (Figure~A.2).

The results show that QPS initially increases as $\delta$ grows, reaches a peak, and then gradually declines. The optimal performance is typically observed when $\delta$ lies between 0.04 and 0.06. The best $\delta$ for $\delta\text{-}\kwnospace{EMG}$ and $\delta\text{-}\kwnospace{EMQG}$ is nearly the same. This finding suggests that introducing a small but non-zero $\delta$ improves the structural connectivity of the graph and facilitates more efficient navigation during search. However, excessively large $\delta$ values lead to overly dense graphs, which in turn increase the traversal cost and reduce query efficiency.

\begin{figure}[t]\vspace{-1mm}
	\centering
	\includegraphics[width=0.6\columnwidth]{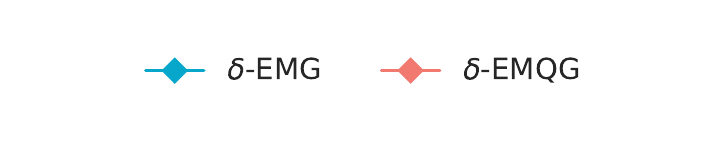}%
	\vspace{-0.9em} 
	
	\captionsetup[subfigure]{justification=centering} 
	\begin{subfigure}[b]{0.50\columnwidth}
		\centering
		\includegraphics[width=\linewidth]{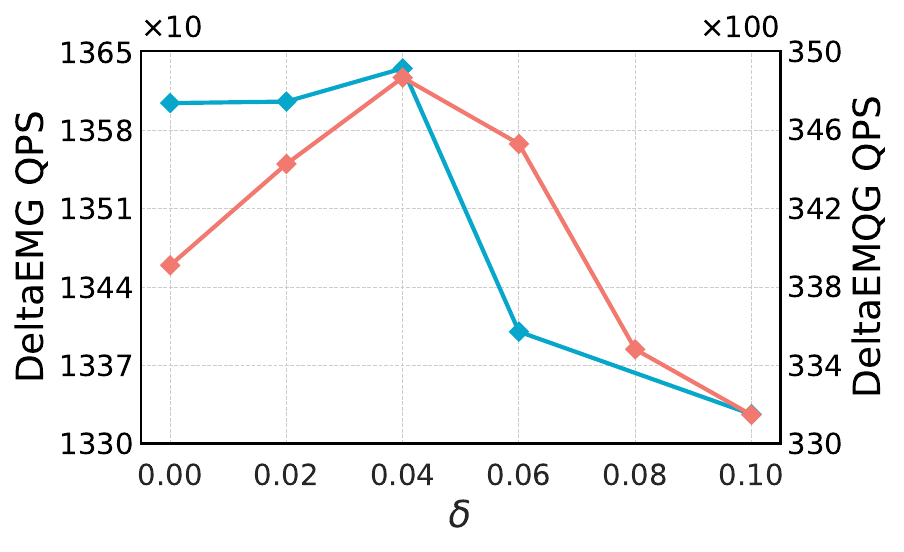}
		\caption{SIFT1M, k=10, recall=95\%}
	\end{subfigure}
	\hspace{-0.02\columnwidth}
	\begin{subfigure}[b]{0.50\columnwidth}
		\centering
		\includegraphics[width=\linewidth]{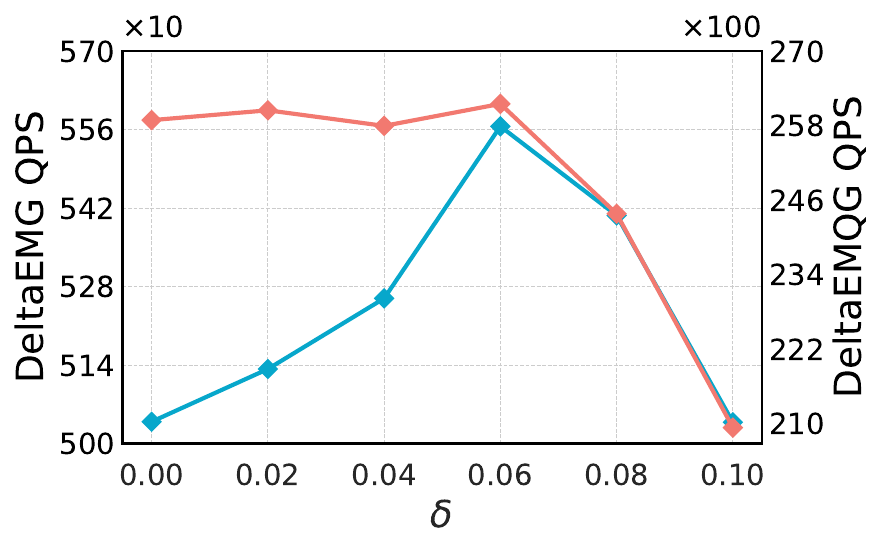}
		\caption{Deep1M, k=10, recall=95\%}
	\end{subfigure}
	\vspace{-6mm}
	\caption{Effect of $\delta$}
	\vspace{-3mm}
	\label{fig:delta-sens-part}
\end{figure}

\stitle{Exp-4. Effect of $t$.}
In our practical construction method, the parameter $t$ controls our adaptive $\delta$ rule, $\delta_t(u, v) = 1 - d(u,v)/d(u,v_{(t)})$, which is designed to balance graph density and search performance. We tested how different values of $t$ affect the QPS of both $\delta\text{-}\kwnospace{EMG}$ and $\delta\text{-}\kwnospace{EMQG}$.

We built indexes with various $t$ values while keeping other parameters constant, and then measured the QPS at 95\% recall with $k=10$. As shown in Figure~\ref{fig:t-sens-part} (see Figure~A.3 in appendix for full results), despite minor fluctuations, the QPS for both methods generally rises to a distinct peak and then declines as $t$ increases. We also noted that the optimal $t$ can differ for the quantized and non-quantized methods. This is likely attributable to the degree alignment procedure in the $\delta\text{-}\kwnospace{EMQG}$ construction, where $t$ is locally adjusted for nodes to ensure their out-degree is a multiple of the SIMD batch size. This additional constraint naturally shifts the optimal value for the $t$ parameter.

Crucially, the peak QPS achieved with the best adaptive $t$ was higher than the best performance attainable with any fixed $\delta$ in the previous experiment. This confirms that our adaptive rule is a more effective approach for building high-performance search graphs.

\begin{figure}[t]\vspace{-3mm}
	\centering
	\includegraphics[width=0.6\columnwidth]{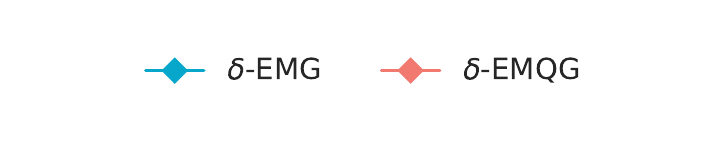}%
	\vspace{-0.9em} 
	
	\captionsetup[subfigure]{justification=centering} 
	\begin{subfigure}[b]{0.50\columnwidth}
		\centering
		\includegraphics[width=\linewidth]{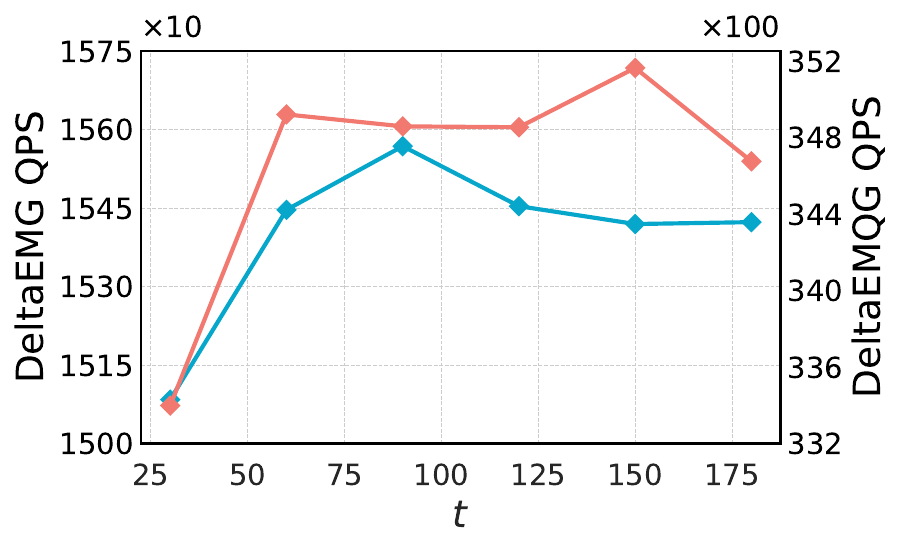}
		\caption{SIFT1M, k=10, recall=95\%}
	\end{subfigure}
	\hspace{-0.02\columnwidth}
	\begin{subfigure}[b]{0.50\columnwidth}
		\centering
		\includegraphics[width=\linewidth]{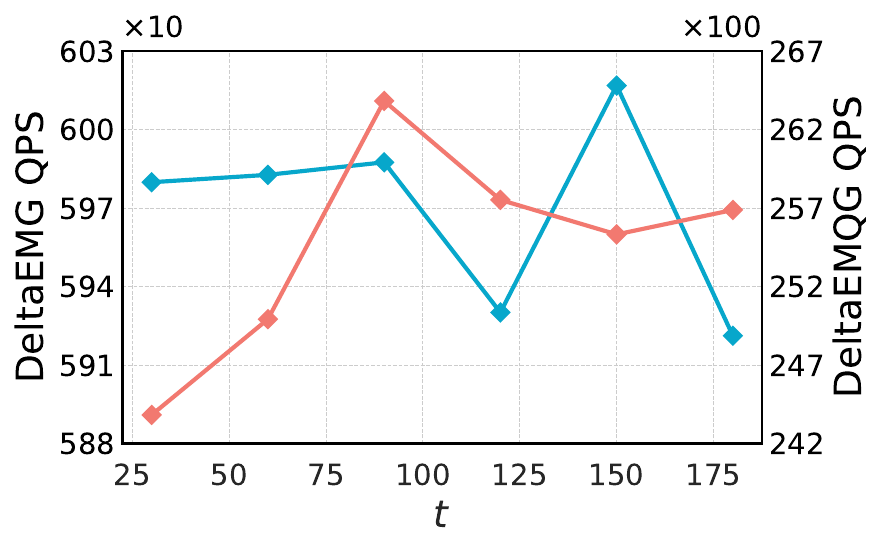}
		\caption{Deep1M, k=10, recall=95\%}
	\end{subfigure}
	\vspace{-6mm}
	\caption{Effect of $t$}
	\vspace{-3mm}
	\label{fig:t-sens-part}
\end{figure}

\subsection{Validating the Error-Bounded Framework}
\label{sec:effect}

\stitle{Exp-5. Analyzing Practical Search Error.}
Here, we test the practical impact of our error-bounded design. To quantify the precision of search results beyond recall, we employ the Relative Distance Error, which is defined as the average error $(d(q,r_{(i)})-d(q,v_{(i)}))/d(q,v_{(i)})$ across all queries and returned vectors.
This metric is intrinsically linked to our theoretical foundation; the result from $\delta$-Error-bounded ANN search guarantees an Relative Distance Error of less than $(1/\delta)-1$.
To provide a fair comparison of algorithmic efficiency, independent of implementation-specific factors like memory layout or caching, we use the average number of distance computations per query as the primary performance metric. This metric directly reflects the core workload of an ANN search. Our analysis focuses on a comparison between $\delta\text{-}\kwnospace{EMG}$ and other non-quantized baselines. We exclude quantized methods from this specific experiment because their performance is governed by a complex trade-off between approximate distance calculations and exact ones, which would obscure the fundamental graph navigation efficiency we aim to measure.

Figure~\ref{fig:distc-rderr-part} illustrates the performance curves on different datasets. Full results are in the appendix (Figure~A.4). We focus our analysis on the high-precision region (Relative Distance Error < 0.005), which typically corresponds to recall rates exceeding 95\%. The results show that $\delta\text{-}\kwnospace{EMG}$ consistently requires fewer distance computations to achieve the same Relative Distance Error compared to all baselines. For example, at an error of 0.001, $\delta\text{-}\kwnospace{EMG}$ needs only 50\% to 80\% of the computations of the next best method. 
This means the geometric guarantees of $\delta\text{-}\kwnospace{EMG}$ make the search process inherently more efficient. It finds accurate neighbors by exploring a smaller portion of the graph, which directly explains the higher QPS we observed in earlier experiments.

\begin{figure}[t]
	\centering
	\includegraphics[width=0.73\columnwidth]{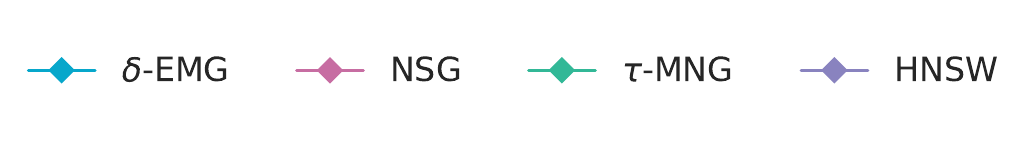}%
	\vspace{-0.2em} 
	
	\captionsetup[subfigure]{justification=centering} 
	\begin{subfigure}[b]{0.503\columnwidth}
		\centering
		\includegraphics[width=\linewidth]{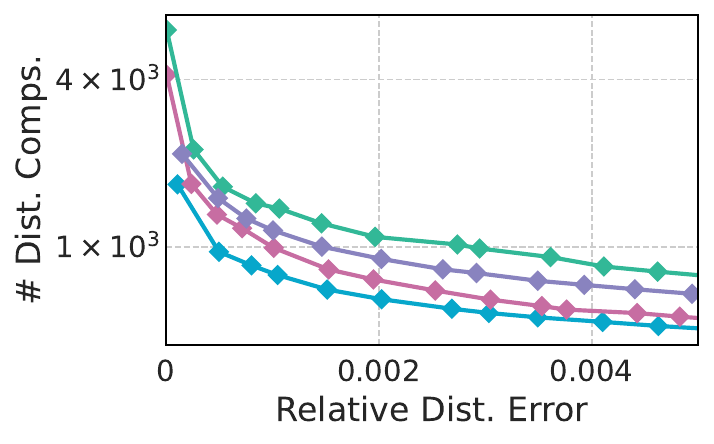}
		\caption{SIFT1M, k=1}
	\end{subfigure}
	\hspace{-0.025\columnwidth}
	\begin{subfigure}[b]{0.503\columnwidth}
		\centering
		\includegraphics[width=\linewidth]{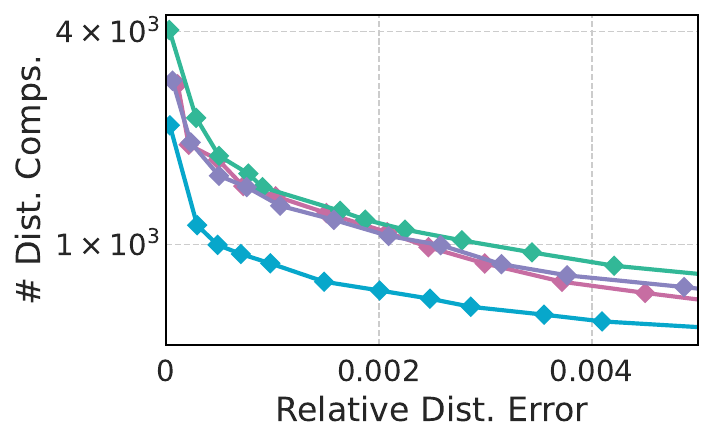}
		\caption{Deep1M, k=1}
	\end{subfigure}
	\vspace{-7mm}
	\caption{Distance Computations vs Relative Distance Error}
	\vspace{-3mm}
	\label{fig:distc-rderr-part}
\end{figure}

\stitle{Exp-6. Probability of Finding a Local Optimum node.}
Our theoretical guarantee for k-ANN search  (Theorem~\ref{thm:kann_guarantee}) depends on finding a "local optimum node" (a node with no closer neighbors) within the final search candidate $C[k:l]$. In this experiment, we test how often this happens.

To clearly validate the theory and analyze the relationship between $\delta$ and its derived $\delta'$, we constructed $\delta\text{-}\kwnospace{EMG}$ graphs using the optimal fixed $\delta$ for each dataset identified in Exp-3 (e.g., $\delta=0.04$ for Crawl, GIST, SIFT1M, and MSong; $\delta=0.06$ for Deep1M). We then performed searches with $k=10$ and varied the search parameter $\alpha$. For each value of $\alpha$, we measured the empirical probability that at least one local optimum exists within the candidate set $C[k:l]$ at the moment of termination.
As shown in Figure~\ref{fig:local_opt_prob}, this probability quickly rises with $\alpha$, approaching 95\% for $\alpha$ around 2.0. This confirms that the condition for our theoretical guarantee is almost always met in practice.

\begin{figure}[t]\vspace{-5mm}
	\centering
	\includegraphics[width=0.63\columnwidth]{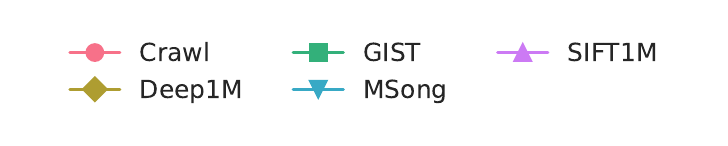}%
	\vspace{-3mm} 
	
	\captionsetup[subfigure]{justification=centering} 
	\begin{subfigure}[b]{0.506\columnwidth}
		\centering
		\includegraphics[width=\linewidth]{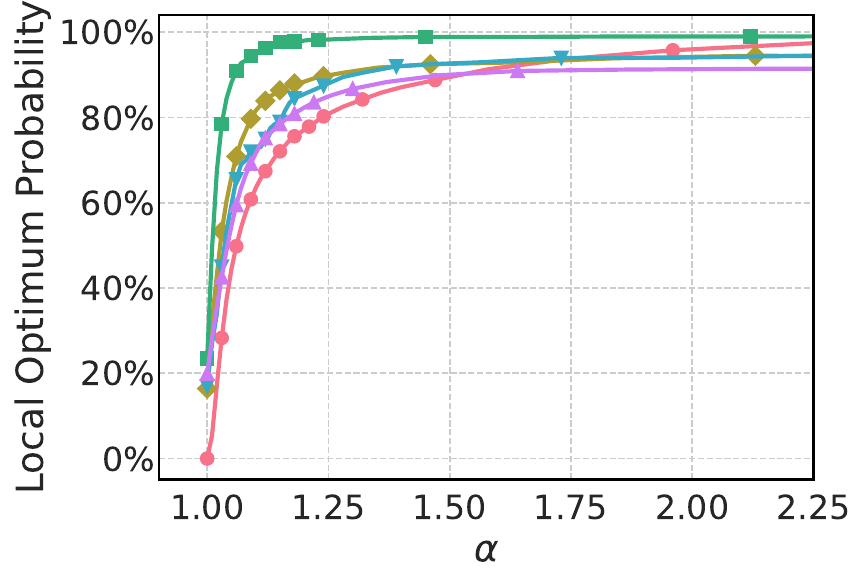}
		\caption{Probability of Local Optimum}
			\label{fig:local_opt_prob}
	\end{subfigure}
	\hspace{-0.03\columnwidth}
	\begin{subfigure}[b]{0.503\columnwidth}
		\centering
		\includegraphics[width=\linewidth]{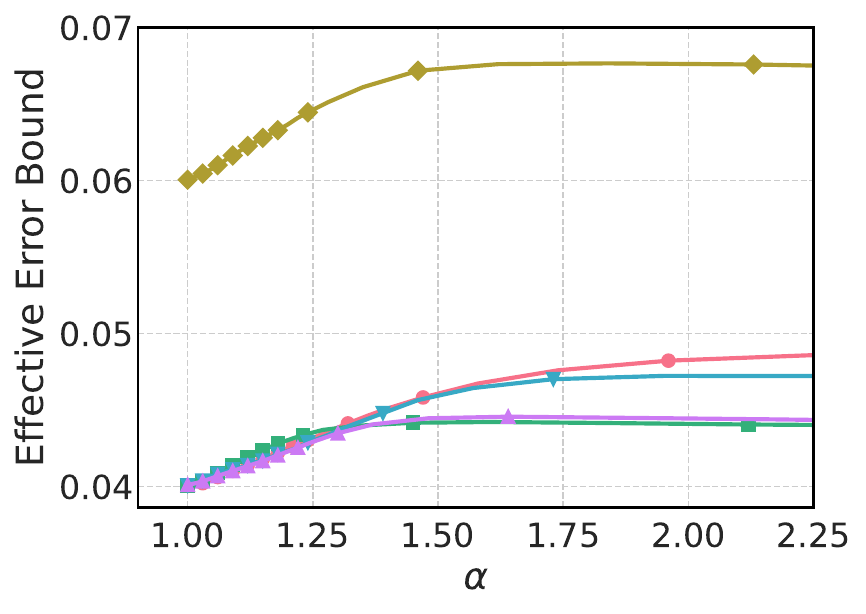}
		\caption{Achieved Error Bound}
		\label{fig:achieved_delta}
	\end{subfigure}
	\vspace{-7mm}
	\caption{Empirical validation of the ANN error guarantee}
	\vspace{-3mm}
	\label{fig:error_bound}
\end{figure}

\stitle{Exp-7. Analyzing the Theoretical Error Bound $\delta'$.}
When a local optimum node is found, the achieved bound is given by $\delta' = \delta \cdot d(q,u)/d(q,r_{(k)})$, where $u$ is the discovered local optimum. A larger $\delta'$ corresponds to a tighter error guarantee. Using the same experimental setup as above, we measured the average $\delta'$ achieved for queries where a local optimum node was found. The results are presented in Figure~\ref{fig:achieved_delta}.

The results show that the average $\delta'$ increases with $\alpha$ and then levels off. This is because a larger $\alpha$ leads to a wider search, which tends to find local optimum nodes that are farther away, improving the bound.

Importantly, the achieved $\delta'$ is always better than the base $\delta$ used to build the graph. For graphs built with $\delta=0.04$, we achieved a $\delta'$ of about 0.045–0.050. For the graph with $\delta=0.06$, we achieved a $\delta'$ of 0.068. This shows our method delivers a practical error guarantee that is even stronger than the one used for construction.

\subsection{Scalability and Ablation Experiments}
\stitle{Exp-8. Scalability Studies.}
In this section, we evaluate how our methods perform as the dataset size increases. We evaluated $\delta\text{-}\kwnospace{EMG}$ and $\delta\text{-}\kwnospace{EMQG}$ on subsets of the SIFT1B dataset of increasing magnitude, specifically with 1M, 5M, 10M, 20M, 50M, and 100M vectors.
We kept the build parameters for $\delta\text{-}\kwnospace{EMG}$ and $\delta\text{-}\kwnospace{EMQG}$ consistent across all scales. The parameter $L$ was fixed at 1000, and $M$ was set to 64 (except for $\delta\text{-}\kwnospace{EMQG}$ on the 50M and 100M datasets where it was reduced to 32 to reduce memory cost without compromising performance). The parameter $t$ was fixed to the best value found on SIFT1M.
We measured the search time needed to reach 95\% recall for $k \in \{1, 10, 100\}$. As shown in Figure~\ref{fig:scalabity}, the search time for both of our methods grows nearly linearly with the number of data points. This result confirms that our methods scale efficiently and are suitable for large-scale applications.

\begin{figure}[t]\vspace{-4mm}
	\centering
	\includegraphics[width=0.78\columnwidth]{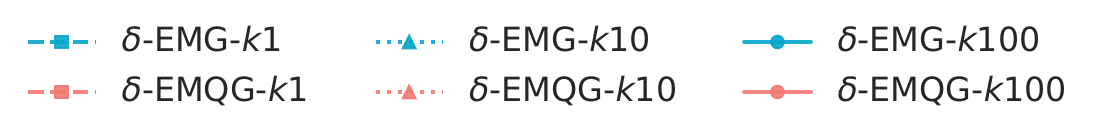}
	\vspace{-0.1em}
	
	\includegraphics[width=0.58\columnwidth,clip,trim=0 0 0 0]{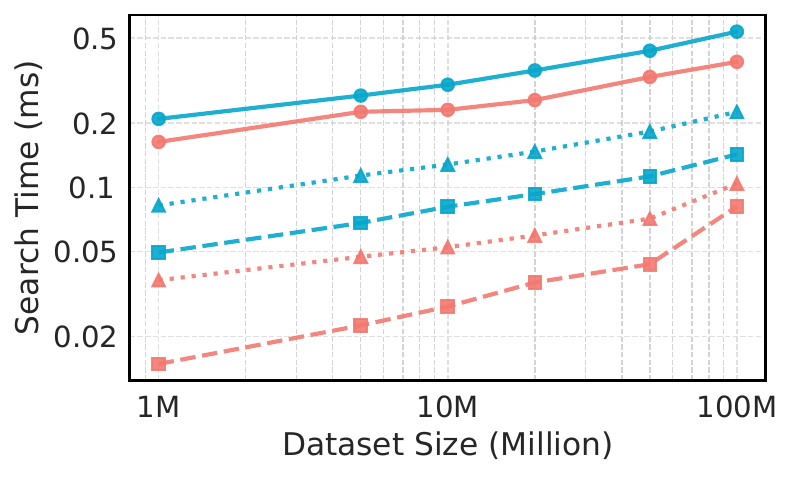}
	\vspace{-4mm}
	\caption{Search Time vs Data Size}
	\label{fig:scalabity}
\end{figure}

\stitle{Exp-9. Ablation Studies.}
To deconstruct the performance contributions of our proposed framework, we conducted an ablation study. We tested the importance of both our graph construction and our search algorithms separately.

First, we kept our specialized search algorithms but replaced our graph construction with baseline structures:
\begin{itemize}[leftmargin=*]
	\item \textbf{$\delta$-EMG-NSG}: Employs our Error-Bounded top-k ANN Search (Algorithm~\ref{alg:eb_ann_search}) on a standard NSG graph.
	
	\item \textbf{$\delta$-EMQG-NSG}: Employs our Probing Search on a SymphonyQG graph, which is based on NSG.
\end{itemize}

Next, we used our proposed graph constructions but reverted to simpler, standard search algorithms:
\begin{itemize}[leftmargin=*]
	\item \textbf{$\delta$-EMG-GS}: Employs a standard greedy search (Algorithm~\ref{alg:greedy_search}) on our $\delta\text{-}\kwnospace{EMG}$ graph.
	
	\item \textbf{$\delta$-EMQG-AGS}: Employs a approximate greedy search (AGS) on our $\delta\text{-}\kwnospace{EMQG}$. Originally proposed in SymphonyQG, AGS follows the same logic as a greedy search but uses approximate distances to guide the graph traversal.
\end{itemize}

We compared the QPS-Recall curves for these variants against our full methods at $k=10$. For the variants with standard search algorithms ($\delta$-EMG-GS and $\delta$-EMQG-AGS), the curves were generated by varying the candidate set size. As shown in Figure~\ref{fig:ablation-part}, removing either our graph construction or our search algorithm leads to a drop in performance. Full results can be seen in the supplementary material. This proves that both components are essential and work synergistically to achieve the final results.

\begin{figure}[t]\vspace{-2mm}
	\centering
	\includegraphics[width=0.74\columnwidth]{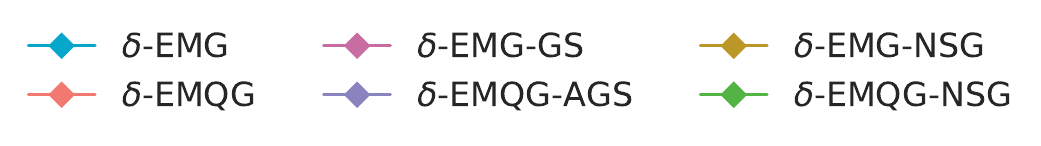}%
	\vspace{-0.2em} 
	
	\captionsetup[subfigure]{justification=centering} 
	\begin{subfigure}[b]{0.50\columnwidth}
		\centering
		\includegraphics[width=\linewidth]{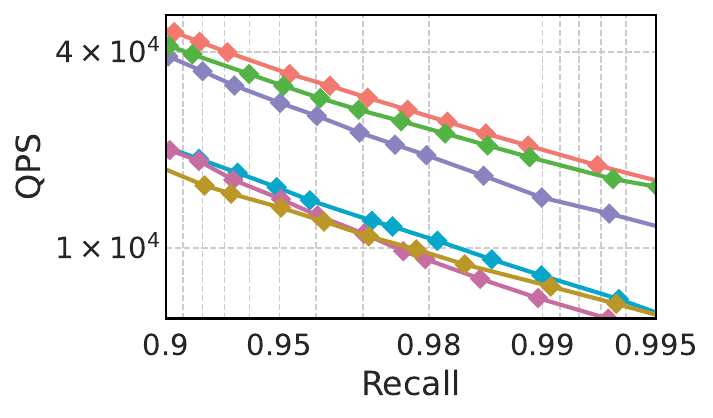}
		\caption{SIFT1M, k=10}
	\end{subfigure}
	\hspace{-0.03\columnwidth}
	\begin{subfigure}[b]{0.50\columnwidth}
		\centering
		\includegraphics[width=\linewidth]{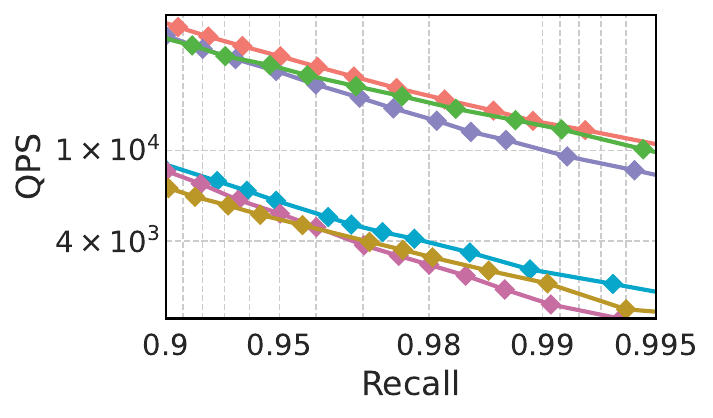}
		\caption{Deep1M, k=10}
	\end{subfigure}
	\vspace{-3mm}
	\caption{Ablation Study}
	\vspace{-3mm}
	\label{fig:ablation-part}
\end{figure}

\section{Conclusion}
In this paper, we addressed the challenge of achieving efficient and error-bounded top-\(k\) Approximate Nearest Neighbor (ANN) search. We proposed a novel graph-based model, the \(\delta\)-error-bounded monotonic graph (\(\delta\)-EMG), which ensures that the distances of retrieved top-\(k\) approximate nearest neighbors are within a user-specified error tolerance. Our method achieves a \((1/\delta)\)-approximation for top-\(k\) NN search with a space complexity of \(O(n \ln n)\). We also introduced \(\delta\)-EMQG, a quantized version of \(\delta\)-EMG, along with efficient construction and search algorithms. Our experimental results on multiple datasets demonstrated that our approach outperforms state-of-the-art methods in terms of query-per-second (QPS) while maintaining the desired error bounds. Furthermore, we provided extensive experiments, including scalability and ablation studies, to validate the effectiveness and efficiency of our algorithms. 

\clearpage

\balance
\bibliographystyle{ACM-Reference-Format}
\bibliography{software}

\end{document}